\documentclass[11pt]{article}

\pdfoutput=1

\usepackage{amsmath,amsthm,amssymb,mathtools} %
\usepackage{authblk} %
\usepackage[english]{babel} %
\usepackage{cite} % references in numerical order
\usepackage{color} %
\usepackage{csquotes} %
\usepackage{dsfont} %
\usepackage[breaklinks]{hyperref} %
\usepackage[hmargin=1in,vmargin=1in]{geometry} %
\usepackage{graphicx} %
\usepackage{hyphenat} %
\usepackage{mathpazo} % charter, fourier, mathpazo, times
\usepackage{mdframed} %
\usepackage{suffix} % for *-version commands
\usepackage{tikz} %
\usepackage{xspace} %

\hypersetup{colorlinks=true, linkcolor=blue, citecolor=magenta}

\usetikzlibrary{calc,positioning}

\mdfdefinestyle{figstyle}{ %
  linecolor=black!7, %
  backgroundcolor=black!7, %
  innertopmargin=10pt, %
  innerleftmargin=25pt, %
  innerrightmargin=25pt, %
  innerbottommargin=10pt %
}

\definecolor{White}{rgb}{1,1,1} %
\definecolor{Black}{rgb}{0,0,0} %
\definecolor{LightGray}{rgb}{.8,.8,.8} %
\colorlet{ChannelColor}{LightGray} %
\colorlet{ChannelTextColor}{Black} %
\colorlet{ReadoutColor}{White} %

\newtheorem{theorem}{Theorem} %
\newtheorem{lemma}[theorem]{Lemma} %
\theoremstyle{definition} %
\newtheorem{definition}[theorem]{Definition} %
\theoremstyle{remark} %
\newtheorem{remark}{Remark} %

\newcommand{\ket}[1]{\ensuremath{\lvert #1 \rangle}} %
\newcommand{\bra}[1]{\ensuremath{\langle #1 \rvert}} %
\newcommand{\ip}[2]{\ensuremath{\left\langle#1,#2\right\rangle}} %
\newcommand{\norm}[1]{\ensuremath{\left\lVert #1 \right\rVert}} %
\newcommand{\abs}[1]{\ensuremath{\left\lvert #1 \right\rvert}} %

\newcommand{\complex}{\mathbb{C}} %
\renewcommand{\natural}{\mathbb{N}} %

\newcommand{\class}[1]{\textup{#1}} %
\newcommand{\reg}[1]{\textsf{#1}} %

\newcommand{\alg}[1]{{\sf #1}\xspace} %

\newcommand{\setft}[1]{\mathrm{#1}} %
\newcommand{\Lin}{\setft{L}} %

\def\X{\mathcal{X}} %
\def\Y{\mathcal{Y}} %
\def\Z{\mathcal{Z}} %
\def\W{\mathcal{W}} %
\def\D{\mathcal{D}} %

\def\I{\mathbb{1}} %

\def\yes{\text{yes}} %
\def\no{\text{no}} %

\DeclareMathOperator{\tr}{Tr} %

\newenvironment{mylist}[1]{\begin{list}{}{ %
      \setlength{\leftmargin}{#1} %
      \setlength{\rightmargin}{0mm} %
      \setlength{\labelsep}{2mm} %
      \setlength{\labelwidth}{8mm} %
      \setlength{\itemsep}{0mm}}}{\end{list}}

\begin{document}

%% End-Of-Header

\title{\LARGE\bf Zero-knowledge proof systems for \class{QMA}}

\author[1]{Anne Broadbent}
\author[2,3]{Zhengfeng Ji}
\author[4]{Fang Song}
\author[5,6]{John Watrous}

\affil[1]{Department of Mathematics and Statistics\protect\\
  University of Ottawa, Canada\vspace{2mm}}

\affil[2]{Centre for Quantum Computation and Intelligent Systems, School of
  Software\protect\\
  Faculty of Engineering and Information Technology\protect\\
  University of Technology Sydney, Australia\vspace{2mm}}

\affil[3]{State Key Laboratory of Computer Science,
  Institute of Software\protect\\ Chinese Academy of Sciences,
  China\vspace{2mm}}

\affil[4]{Institute for Quantum Computing and Department of
  Combinatorics \& Optimization\protect\\
  University of Waterloo, Canada\vspace{2mm}}

\affil[5]{Institute for Quantum Computing and School of Computer
  Science\protect\\
  University of Waterloo, Canada\vspace{2mm}}

\affil[6]{Canadian Institute for Advanced Research\protect\\
  Toronto, Canada}

\date{\today}

\renewcommand\Affilfont{\normalsize\itshape}
\renewcommand\Authfont{\large}
\setlength{\affilsep}{6mm}
\renewcommand\Authsep{\rule{10mm}{0mm}}
\renewcommand\Authands{\rule{10mm}{0mm}}

\maketitle
\setcounter{page}{0}

\thispagestyle{empty}

\begin{abstract}
  Prior work has established that all problems in \class{NP} admit classical
  zero-knowledge proof systems, and under reasonable hardness assumptions for
  quantum computations, these proof systems can be made secure against quantum
  attacks.
  We prove a result representing a further quantum generalization of this fact,
  which is that every problem in the complexity class \class{QMA} has a quantum
  zero-knowledge proof system.
  More specifically, assuming the existence of an unconditionally binding and
  quantum computationally concealing commitment scheme, we prove that every
  problem in the complexity class \class{QMA} has a quantum interactive proof
  system that is zero-knowledge with respect to efficient quantum computations.

  Our \class{QMA} proof system is sound against arbitrary quantum provers, but
  only requires an honest prover to perform polynomial-time quantum
  computations, provided that it holds a quantum witness for a given instance of
  the \class{QMA} problem under consideration.
  The proof system relies on a new variant of the \class{QMA}-complete local
  Hamiltonian problem in which the local terms are described by Clifford
  operations and standard basis measurements.
  We believe that the QMA-completeness of this problem may have other uses
  in quantum complexity.
\end{abstract}

\newpage

%------------------------------------------------------------------------------%
\section{Introduction}
\label{sec:intro}
%------------------------------------------------------------------------------%

Zero-knowledge proof systems, first introduced by Goldwasser, Micali and
Rackoff~\cite{GMR89}, are interactive protocols that allow a prover to convince
a verifier of the validity of a statement while revealing no additional
information beyond the statement's validity.
Although paradoxical as it appears, several problems that are not known to
be efficiently computable, such as the Quadratic Non-Residuosity, Graph
Isomorphism, and Graph Non-Isomorphism problems, were shown to admit
zero-knowledge proof systems~\cite{GMR89,GMW91}.
Under reasonable intractability assumptions, Goldreich, Micali and
Wigderson~\cite{GMW91} gave a zero-knowledge protocol for the Graph
$3$-Coloring problem and, because of its \class{NP}-completeness, for all
\class{NP} problems.
This line of work was further extended in~\cite{BOGG+90}, which showed that all
problems in \class{IP} have zero-knowledge proof systems.

Since the invention of this concept, zero-knowledge proof systems have become a
cornerstone of modern theoretical cryptography.
In addition to the conceptual innovation of formulating a complexity-theoretic
notion of knowledge, zero-knowledge proof systems are essential building blocks
in a host of cryptographic constructions.
One notable example is the design of secure two-party and multi-party
computation protocols~\cite{GMW87}.

The extensive works on zero-knowledge largely reside in a classical world.
The development of quantum information science and technology has urged
another look at the landscape of zero-knowledge proof systems in a
\emph{quantum} world.
Namely, both honest users and adversaries may potentially possess the
capability to exchange and process quantum information.
There are, of course, zero-knowledge protocols that immediately become insecure
in the presence of quantum attacks due to efficient quantum algorithms that
break the intractability assumptions upon which these protocols rely.
For instance, Shor's quantum algorithms for factoring and computing discrete
logarithms~\cite{Sho97} invalidate the use of these problems, generally
conjectured to be classically hard, as a basis for the security of
zero-knowledge protocols against quantum attacks.
Even with computational assumptions against quantum adversaries, however, it
is still highly nontrivial to establish the security of classical
zero-knowledge proof systems in the presence of malicious \emph{quantum}
verifiers because of a technical reason that we now briefly explain.

The zero-knowledge property of a proof system for a fixed input string is
concerned with the computations that may be realized through an interaction
between a (possibly malicious) verifier and the prover.
That is, the malicious verifier may take an arbitrary input (usually
called the \emph{auxiliary input} to distinguish it from the input string to
the proof system under consideration), interact with the prover in any way it
sees fit, and produce an output that is representative of what it has learned
through the interaction.
Roughly speaking, the prover is said to be \emph{zero-knowledge} on the
fixed input string if any computation of the sort just described can be
efficiently approximated\footnote{
  Different notions of approximations are considered, including
  \emph{statistical} approximations and \emph{computational} approximations,
  which require that the simulator's computation is either statistically
  (or information-theoretically) indistinguishable or computationally
  indistinguishable from the malicious verifier's computation.
  This paper is primarily concerned with the computational variant.
}
by a \emph{simulator} operating entirely on its own---meaning that it does not
interact with the prover, and in the case of an \class{NP} problem it does not
possess a witness for the fixed problem instance being considered.
The proof system is then said to be zero-knowledge when this zero-knowledge
property holds for all yes-instances of the problem under consideration.

Classically speaking, the zero-knowledge property is typically established
through a technique known as \emph{rewinding}.
In essence, the simulator can store a copy of its auxiliary input, and it can
make guesses and store intermediate states representing a hypothetical
prover/verifier interaction---and if it makes a bad guess or otherwise
experiences bad luck when simulating this hypothetical interaction,
it simply reverts to an earlier stage (or possibly back to the beginning) of
the simulation and tries again.
Indeed, it is generally the simulator's freedom to disregard the temporal
restrictions of the actual prover/verifier interaction in a way such as this
that makes it possible to succeed.

However, rewinding a quantum simulation is more problematic; the
\emph{no-cloning theorem}~\cite{WZ82} forbids one from copying quantum
information, making it impossible to store a copy of the input or of an
intermediate state, and measurements generally have an irreversible
effect~\cite{FP96} that may partially destroy quantum information.
Such difficulties were first observed by van de Graaf~\cite{Gra97} and further
studied in~\cite{Wat02,DFS04}.
Later, a \emph{quantum rewinding} technique was found~\cite{Wat09} to establish
that several interactive proof systems, including the
Goldreich-Micali-Wigderson Graph $3$-Coloring proof system~\cite{GMW91},
remain zero-knowledge against malicious quantum verifiers
(under appropriate quantum intractability assumptions in some cases).
It follows that all \class{NP} problems have zero-knowledge proof systems even
against quantum malicious verifiers, provided that a quantum analogue of the
intractability assumption required by the
Goldreich-Micali-Wigderson Graph $3$-Coloring proof system are in place.

This work studies the quantum analogue of \class{NP}, known as \class{QMA},
in the context of zero-knowledge.
These are problems with a succinct \emph{quantum} witness satisfying similar
completeness and soundness to \class{NP} (or its randomized variant
\class{MA}).
Quantum witnesses and verification are conjectured to be more powerful than
their classical counterparts: there are problems that admit short quantum
witnesses, whereas there is no known method for verification using a
polynomial-sized classical witness.
In other words, $\class{NP}\subseteq\class{QMA}$ holds trivially, and the
containment is typically conjectured to be proper.
The question we address in this paper is:
\emph{Does every problem in \class{QMA} have a zero-knowledge quantum
  interactive proof system?}
In more philosophical terms, viewing quantum witnesses as precious sources of
knowledge:
\emph{Can one always devise a proof system that reveals nothing about
a quantum witness beyond its validity?}

\subsection{Our contributions}

We answer the above question positively by constructing a quantum interactive
proof system for any problem in \class{QMA} that is zero-knowledge against
any polynomial-time quantum adversary, under a reasonable quantum
intractability assumption.

\begin{theorem}
  \label{thm:main}
  Assuming the existence of an unconditionally binding and quantum
  computationally concealing bit commitment scheme, every problem in
  \class{QMA} has a quantum computational zero-knowledge proof system.
\end{theorem}

\noindent
A few of the desirable features of our proof system are as follows:

\begin{mylist}{\parindent}
\item[1.]
  Our proof system has a simple structure, similar to the classical
  Goldreich-Micali-Wigderson Graph $3$-Coloring proof system (and to
  the so-called $\Sigma$-protocols more generally).
  It can be viewed as a three-phase process: the prover commits to a quantum
  witness, the verifier makes a random challenge, and finally the prover
  responds to the challenge by partial opening of the committed information
  that suffices to certify the validity.

\item[2.]
  All communications in our proof system are classical except for the first
  commitment message, and the verifier can measure the quantum message
  immediately upon its arrival (which has a strong technological appeal).

\item[3.]
  Our protocol is based on mild computational assumptions.
  The sort of bit commitment scheme it requires can be implemented, for
  instance, under the existence of injective one-way functions that are hard to
  invert in quantum polynomial time.

\item[4.]
  Our protocol is prover-efficient.
  It is sound against general quantum provers, but given a valid quantum
  witness, an honest prover only needs to perform efficient quantum
  computations.
  As has already been suggested, aside from the preparation of the first
  quantum message, all of the remaining computations performed by the honest
  prover are classical polynomial-time computations.

\end{mylist}

As a key ingredient of our zero-knowledge proof system, we introduce a new
variant of the \mbox{$k$-local} Hamiltonian problem and prove that it remains
\class{QMA}-complete (with respect to Karp reductions).
The $k$-local Hamiltonian problem asks if the minimum eigenvalue
(or ground state energy in physics parlance) of an $n$-qubit Hamiltonian
$H=\sum_j H_j$, where each $H_j$ is $k$-local
(i.e., acts trivially on all but $k$ of the $n$ qubits), is below a particular
threshold value.
This problem was introduced and proved to be \class{QMA}-complete (for the case
$k=5$) by Kitaev~\cite{KSV02}.
We show that each $H_j$ can be restricted to be realized by a
Clifford operation, followed by a standard basis measurement, and the
\class{QMA}-completeness is preserved.
Beyond its use in this paper, this fact has the potential to provide other
insights into the study of quantum Hamiltonian complexity.
For an arbitrary problem $A\in \class{QMA}$, we can reduce an instance of
$A$ efficiently to an instance of the $k$-local Clifford Hamiltonian problem, and
a valid witness for $A$ can also be transformed into a witness
for the corresponding $k$-local Clifford Hamiltonian problem instance
by an efficient quantum procedure.
As a result, $A$ has a zero-knowledge proof system by composing this reduction
with our zero-knowledge proof system for the $k$-local Clifford Hamiltonian
problem.

Our proof system also employs a new encoding scheme for quantum states, which we
construct by extending the \emph{trap scheme} proposed in~\cite{BGS13}.
While our new scheme can be seen as a \emph{quantum authentication scheme}
(cf.~\cite{BCG+02,BCG+06,ABE10}), it in addition allows performing arbitrary
constant-qubit Clifford circuits and measuring in the computational basis
directly on authenticated data without the need for auxiliary states.
Previously the only known scheme supporting this feature requires
high-dimensional quantum systems (i.e., qudits rather than
qubits)~\cite{BCG+06}, which make it inconvenient in our setting where all
quantum operations are on qubits.

\subsection{Overview of protocol and techniques}

A natural approach to constructing zero-knowledge proofs for~\class{QMA} is
to consider a quantum analogue of the Goldreich-Micali-Wigderson proof system
for Graph $3$-Coloring (which we will hereafter refer to as the GMW $3$-Coloring
proof system).
Let us focus in particular on the local Hamiltonian problem, and consider a
proof system in which the prover holds a quantum witness state for an instance
of this problem, commits to this witness, and receives the challenge from the
verifier (which, let us say, is a random term of the local Hamiltonian).
The prover might then open the commitments of the set of qubits on which the
term acts non-trivially so that the verifier can measure the local energy for
this term and determine acceptance accordingly.

There is a major difficulty when one attempts to carry out such an
approach for \class{QMA}.
The zero-knowledge property of the GMW $3$-Coloring proof system depends
crucially on a structural property of the problem: the honest prover is free to
randomize the three colors used in its coloring, and when the commitments
to the colors of two neighboring vertices are revealed, the verifier will see
just a uniform mixture over all pairs of different colors.
This uniformity of the coloring marginals is important in achieving the
zero-knowledge property of the proof system.
Unlike the case of $3$-Coloring, however, none of the known \class{QMA}-complete
problems under Karp reductions has such desirable properties.
For example, if we use local Hamiltonian problems directly in a GMW-type
proof system, of the sort suggested above, information about the reduced state
of the quantum witness will be leaked to the verifier, possibly violating the
zero-knowledge requirement.

To overcome the difficulty suggested above, we employ several ideas that enable
the prover to ``partially'' open the commitments, revealing only the fact that
the committed state lives in certain subspaces, and nothing further.
Our first technique simplifies the verification circuit for
\class{QMA}-complete problems through the introduction of the local
Clifford-Hamiltonian problem that was already described.
Somewhat more specifically, our formulation of this problem requires every
Hamiltonian term to take the form $C^{\ast} \ket{0^k}\bra{0^k} C$ for some
Clifford operation~$C$.
Because the local Clifford-Hamiltonian problem remains \class{QMA}-complete, it
implies a random Clifford verification procedure for problems
in~\class{QMA}: intuitively, the verification of a quantum witness has been
simplified to a Clifford measurement followed by a classical verification.

The Clifford verification procedure works in harmony with the encryption of
quantum data via the quantum one-time pad and other derived hybrid schemes
that are used by our proof system.
This has the important effect of transforming statements about quantum states
into those about the classical keys of the quantum one-time pad, which
naturally leads to our second main idea: the use of zero-knowledge
proofs for \class{NP} against quantum attacks to simplify the construction
of zero-knowledge proofs for~\class{QMA}.
In our protocol, the verifier measures the encrypted quantum data and asks
the prover to prove, using a zero-knowledge protocol for \class{NP}, that
the decryption of this result is consistent with the verifier accepting.

In fact, if the verifier measures the quantum data according to the
specifications of the protocol, the combination of the Clifford
verification and the use of zero-knowledge proofs for \class{NP}
suffices.
A problem arises, however, if the verifier does not perform the honest
measurement.
Our third technique, inspired by work on quantum
authentication~\cite{BCG+06,ABE10,DNS12,BGS13}, employs a new scheme for
encoding quantum states.
Roughly speaking, if the prover encodes a witness state under our encoding
scheme, then the verifier is essentially forced to perform the measurement
honestly---any attempt to fake a ``logically different'' measurement result will
succeed with negligible probability.
In our proof system, we adapt the trap scheme proposed in~\cite{BGS13} so that
we can perform any constant-sized Clifford operations on authenticated quantum
data followed by computational basis measurements, benefiting along the way
from ideas concerning quantum computation on authenticated quantum data.

The resulting zero-knowledge proof system for \class{QMA} has a similar
overall structure to the GMW $3$-Coloring protocol: the prover encodes the
quantum witness state using a quantum authentication scheme, and sends the
encoded quantum data together with a commitment to the secret keys of the
authentication to the verifier.
The verifier randomly samples a term $C^{\ast} \ket{0^k}\bra{0^k} C$ in the
local Clifford-Hamiltonian problem, applies the operation $C$ transversally
on the encoded quantum data and measures all qubits corresponding to the $k$
qubits of the selected term in the computational basis, and sends the
measurement outcomes to the prover.
The prover and verifier then invoke a quantum-secure zero-knowledge proof for
the \class{NP} statement that the commitment correctly encodes an
authentication key and, under this key, the verifier's measurement outcomes
do not decode to $0^k$.

\subsection{Comparisons to related work}

There has been other work on quantum complexity and theoretical
cryptography, some of which is discussed below, that allows one to conclude
statements having some similarity to our results.
We will argue, however, that with respect to the problem of devising
zero-knowledge quantum interactive proof systems for \class{QMA}, our main
result is stronger in almost all respects.
In addition, we believe that our proof system is appealing both because it
is conceptually simple and represents a natural extension of well-known
classical methods.

\begin{mylist}{\parindent}
\item[1.] \emph{Zero-knowledge proof systems for all of \class{IP}.}
  Hallgren, Kolla, Sen and Zhang~\cite{HKSZ08} proved that classical
  zero-knowledge proof systems for \class{IP}~\cite{BOGG+90} can be made
  secure against malicious quantum verifiers under a certain technical
  condition.
  It appears that this condition holds assuming the existence of a quantum
  computationally hiding commitment scheme.
  Because \class{QMA} is contained in \class{IP}, this would imply a classical
  zero-knowledge protocol for \class{QMA}.
  However, this generic protocol would require a computationally
  \emph{unbounded} prover to carry out the honest protocol, and it is unlikely
  to reduce the round complexity without causing unexpected consequences in
  complexity theory~\cite{GS86,Wat03,GO94}.

\item[2.] \emph{Secure two-party computations.}
  Another approach to constructing zero-knowledge proofs for
  \class{QMA} is to apply the general tool of secure two-party quantum
  computation~\cite{BCG+06,DNS10,DNS12}.
  In particular, we may imagine two parties, a prover and a verifier, jointly
  evaluating the verification circuit of a \class{QMA} problem, with the
  prover holding a quantum witness as his/her private input.
  In principle, one can design a two-party computation protocol so that the
  verifier learns the validity of the statement but nothing more about the
  prover's private input.
  While we believe that a careful analysis could make this approach work, it
  comes at a steep cost.
  First, we need to make significantly stronger computational assumptions, as
  secure quantum two-party computation relies on (at least) secure computations
  of classical functions against quantum adversaries.
  The best-known quantum-secure protocols for classical two-party computation
  assume quantum-secure dense public-key encryption~\cite{HSS15} or similar
  primitives~\cite{LN11}, in contrast to the existence of a quantum
  computationally hiding commitment scheme.\footnote{
    Roughly speaking, this distinction is analogous to ``Cryptomania'' vs
    ``minicrypt'' according to Impagliazzo's five-world paradigm~\cite{Imp95}.}
  Secondly, the protocol obtained this way is only an \emph{argument} system.
  That is, the protocol is only sound against computationally bounded dishonest
  provers.
  Moreover, the generic quantum two-party computation protocol evaluates the
  verification circuit gate by gate, and in particular interactions are
  unavoidable for some (non-Clifford) gates.
  This causes the round complexity to grow in proportion to the size of the
  verification circuit.
  In addition, the communications are inherently quantum, which makes the
  protocol much more demanding from a technological viewpoint.

  On the positive side, through this approach, it is possible to achieve
  negligible soundness error using just one copy of witness state.
  In contrast, our proof system directly inherits the soundness error of the
  most natural and direct verification for the local Clifford-Hamiltonian
  problem (i.e., randomly select a Hamiltonian term and measure).
  If one reduces an arbitrary \class{QMA}-verification procedure to an instance
  of this problem, the resulting soundness guarantee could be significantly
  worse.

\item[3.] \emph{Zero-knowledge proofs for Density Matrix Consistency.}
  It was pointed out by Liu~\cite{Liu06} that the Density Matrix Consistency
  problem, which asks if there exists a global state of $n$ qubits that is
  consistent with a collection of $k$-qubit density matrix marginals, should
  admit a simple zero-knowledge proof system following the GMW $3$-Coloring
  approach.
  This fact was one of the inspirations for our work.
  While it approaches our main result, it does not necessarily admit a
  zero-knowledge proof system for all problems in \class{QMA}, as the
  Density Matrix Consistency problem is only known to be hard for
  \class{QMA} with respect to Cook reductions.

\item[4.] \emph{Other results on Clifford verifications for \class{QMA}.}
  We note that Clifford verification with classical post-processing of
  \class{QMA} was considered in~\cite{MHNF15} using magic states as
  ancillary resources.
  Our construction is arguably simpler, uses only constant-size Clifford
  operations, and most importantly does not require any resource states.
  This helps to avoid checking the correctness of resource states in the final
  zero-knowledge protocol.
  We are hopeful that our techniques will provide new insights to the study of
  quantum Hamiltonian complexity, and may find useful applications in other
  areas of research such as the study of non-local games.
  One byproduct of our Clifford-Hamiltonian reduction proof is an alternative
  proof of the single-qubit measurement verification for \class{QMA} recently
  proposed by~\cite{MNS16}.
\end{mylist}

\subsubsection*{Organization}
The remainder of the paper is organized as follows.
Section~\ref{sec:LCH} describes the variant of the local Hamiltonian problem
mentioned above.
We present our zero-knowledge proof system for \class{QMA} in
Section~\ref{sec:proof-system-description} and prove its completeness and
soundness in Section~\ref{sec:completeness-and-soundness} and zero-knowledge
property in Section~\ref{sec:zero-knowledge}.
We conclude with some remarks and future directions in Section~\ref{sec:con}.
An appendix summarizing basic notation, definitions, and useful primitives for
the construction of our zero-knowledge proof system is also included for
completeness.

%------------------------------------------------------------------------------%
\section{The local Clifford-Hamiltonian problem}
\label{sec:LCH}
%------------------------------------------------------------------------------%

The local Hamiltonian problem~\cite{KSV02} is a well-known example of a
complete problem for \class{QMA}, provided that certain assumptions are
in place regarding the gap between the ground state energy (i.e., the smallest
eigenvalue) of input Hamiltonians for yes- and no-inputs.
A general and somewhat imprecise formulation of the local Hamiltonian problem
is as follows.
\vspace{2mm}

\noindent
\emph{The $k$-local Hamiltonian problem ($k$-LH)}\vspace{2mm}
\newline
\noindent
\begin{tabular*}{\textwidth}{@{}p{0.4in}@{\hspace*{3mm}}p{5.9in}}
  \emph{Input:} &
  A collection $H_1,\ldots,H_m$ of $k$-local Hamiltonian operators, each
  acting on $n$ qubits and satisfying $0 \leq H_j \leq \I$ for
  $j = 1,\ldots,m$, along with real numbers $\alpha$ and $\beta$ satisfying
  $\alpha < \beta$.\\[2mm]
  \emph{Yes:} &
  There exists an $n$-qubit state $\rho$ such that
  $\ip{\rho}{H_1+\cdots+H_m} \leq \alpha$.\\[2mm]
  \emph{No:} &
  For every $n$-qubit state $\rho$, it holds that
  $\ip{\rho}{H_1+\cdots+H_m} \geq \beta$.
\end{tabular*}
\vspace{2mm}

\noindent
This problem statement is imprecise in the sense that it does not specify how
$\alpha$ and~$\beta$ are to be represented or what requirements are placed on
the gap $\beta - \alpha$ mentioned above.
We will be more precise about these issues when formulating a restricted
version of this problem below, but it is appropriate that we first summarize
what is already known.

It is known that $k$-LH is complete for \class{QMA} (with respect to Karp
reductions) provided $\alpha$ and $\beta$ are input in a reasonable way and
separated by an inverse polynomial gap;
this was first proved by Kitaev~\cite{KSV02} for the case $k = 5$, then by
Kempe and Regev~\cite{KR03} for $k = 3$ and Kempe, Kitaev, and Regev
\cite{KKR06} for $k = 2$.
If one adds the additional requirement that $\alpha$ is exponentially small,
which will be important in the context of this paper, then
$\class{QMA}$-completeness for $k=5$ still follows from Kitaev's proof, but the
proofs of Kempe and Regev and Kempe, Kitaev, and Regev do not imply the same
for $k=3$ and $k=2$.
On the other hand, the work of Bravyi~\cite{Bra11} and Gosset and Nagaj
\cite{GN13}
does establish $\class{QMA}$-completeness for exponentially small $\alpha$, for
$k=4$ and $k=3$, respectively.

The restricted version of the local Hamiltonian we introduce is one in which
each Hamiltonian term $H_j$ is not only $k$-local and satisfies
$0 \leq H_j \leq \I$, but furthermore on the $k$ qubits on which it acts
nontrivially, its action must be given by a rank 1 projection operator of the
form
\begin{equation}
  C_j^{\ast} \ket{0^k}\bra{0^k} C_j,
\end{equation}
for some choice of a $k$-qubit Clifford operation $C_j$.
For brevity, we will refer to any such operator as a
\emph{$k$-local Clifford-Hamiltonian projection}.
The precise statement of our problem variant is as follows.

\noindent
\emph{The $k$-local Clifford-Hamiltonian problem ($k$-LCH)}\vspace{2mm}
\newline
\noindent
\begin{tabular*}{\textwidth}{@{}p{0.4in}@{\hspace{3mm}}p{5.9in}}
  \emph{Input:} &
  A collection $H_1,\ldots,H_m$ of $k$-local Clifford-Hamiltonian
  projections, along with positive integers $p$ and $q$ expressed in unary
  notation (i.e., as strings $1^p$ and $1^q$) and satisfying $2^{p}>q$.
  \\[2mm]
  \emph{Yes:} &
  There exists an $n$-qubit state $\rho$ such that
  $\ip{\rho}{H_1+\cdots+H_m} \leq 2^{-p}$.
  \\[2mm]
  \emph{No:} &
  For every $n$-qubit state $\rho$, it holds that
  $\ip{\rho}{H_1+\cdots+H_m} \geq 1/q$.
\end{tabular*}
\vspace{2mm}

\noindent
It may be noted that, by the particular way we have stated this problem,
we are focusing on a variant of the local Hamiltonian problem in which the
parameter $\alpha$ may be exponentially small and the gap $\beta-\alpha$ is at
least inverse polynomial.

\begin{theorem}
  \label{thm:lch}
  The $5$-local Clifford-Hamiltonian problem is \class{QMA}-complete
  with respect to Karp reductions. Moreover, for any choice of
  promise problem $A = (A_{\yes},A_{\no}) \in \class{QMA}$ and a
  polynomially bounded function $p$, there exists a Karp reduction $f$
  from $A$ to $5$-LCH having the form
  \begin{equation}
    f(x) = \Bigl\langle H_1,\ldots,H_m,1^{p(\abs{x})},1^q\Bigr\rangle
  \end{equation}
  for every $x\in A_{\yes}\cup A_{\no}$.
\end{theorem}

\begin{proof}
  The containment of the 5-local Clifford-Hamiltonian problem in \class{QMA}
  follows from the fact that the 5-LH problem is in \class{QMA} for the
  same choice of the ground state energy bounds.
  It therefore remains to prove the statement concerning the
  \class{QMA}-hardness of the 5-LCH problem.

  Let $A = (A_{\yes}, A_{\no})$ be any promise problem in \class{QMA} and let
  $p$ be a polynomially bounded function.
  Using a standard error reduction procedure for \class{QMA}, one may conclude
  that there exists a polynomial-time generated collection
  $\{V_x\,:\,x\in A_{\yes}\cup A_{\no}\}$ of measurement circuits having
  these properties:
  \begin{mylist}{\parindent}
  \item[1.] If $x\in A_{\yes}$, there exists a state $\rho$ such that
    $V_x(\rho) = 1$ with probability $1-2^{-p(\abs{x})}$.
  \item[2.] If $x\in A_{\no}$, then for all quantum states $\rho$ representing
    valid inputs to $V_x$ it holds that $V_x(\rho) = 1$ with probability at
    most $1/2$.
  \end{mylist}

  It is known that $\{\Lambda(P), H\}$ is a universal gate set for quantum
  computation, so there would be no loss of generality in assuming each
  $V_x$ is a quantum circuit using gates from this set, together with a supply
  of ancillary qubits initialized to the state $\ket{0}$.
  For technical reasons (which are discussed later) we will assume something
  marginally stronger, which is that each $V_x$ uses gates from the set
  $\{\Lambda(P), H\otimes H\}$.
  That is, every Hadamard gate appearing in $V_x$ is paired with another
  Hadamard gate to be applied at the same time but on a different qubit.
  Note that for any circuit composed of gates from the set $\{\Lambda(P), H\}$,
  this stronger condition is easily met by adding to this circuit a number of
  additional Hadamard gates on an otherwise unused ancilla qubit.

  Now consider the $5$-local circuit-to-Hamiltonian construction of
  Kitaev~\cite{KSV02}, for a given choice of $V_x$.
  In this construction, the resulting Hamiltonians have the form
  \begin{equation}
    H_{\text{total}} = H_{\text{in}} + H_{\text{out}} +
    H_{\text{clock}} + H_{\text{prop}},
  \end{equation}
  where the terms check the initialization, readout, validity of unary clock,
  and propagation of computation respectively.
  It follows from Kitaev's proof that, for $x\in A_{\yes}$, the resulting
  Hamiltonian $H_{\text{total}}$ has ground state energy at most
  $\smash{2^{-p(|x|)}}$, and for $x\in A_{\no}$ the ground state energy of
  $H_\text{total}$ is at least~$1/q(|x|)$, for some polynomially bounded
  function $q$.
  To complete the proof, it suffices to demonstrate that each of these terms
  can be expressed as a sum of Clifford-Hamiltonian projections.

  The first three terms, $H_{\text{in}}$, $H_{\text{out}}$, and
  $H_{\text{clock}}$, can be expressed as sums of Clifford-Hamiltonian
  projections easily, as they are all projection operators that are diagonal in
  the standard basis.
  The propagation term has the form
  $H_{\text{prop}} = \sum_{t=1}^T H_{\text{prop},t}$ where
  each operator $H_{\text{prop},t}$ takes the form
  \begin{equation}
    \begin{split}
      H_{\text{prop},t} & = \frac{1}{2} \bigl[
      (\ket{100}\bra{100}_{t-1,t,t+1} +
      \ket{110}\bra{110}_{t-1,t,t+1})\otimes\I\\
      & \qquad -\ket{110}\bra{100}_{t-1,t,t+1} \otimes U_t
      -\ket{100}\bra{110}_{t-1,t,t+1} \otimes U_t^\ast \bigr]\\
      & = \ket{10}\bra{10}_{t-1,t+1} \otimes \frac{1}{2} \bigl[ \I_t
      \otimes \I - \ket{1}\bra{0}_t \otimes U_t - \ket{0}\bra{1}_t
      \otimes U_t^\ast \bigr].
    \end{split}
  \end{equation}
  Here, the first three qubits (indexed by $t-1$, $t$, and $t+1$) refer to
  qubits in a clock register and $U_t$ represents the $t$-th unitary gate in
  $V_x$.
  To prove that each propagation operator $H_{\text{prop},t}$ can be expressed
  as a sum of Clifford-Hamiltonian projections, it suffices to prove the same
  for every projection of the form
  \begin{equation}
    \label{eq:projection-from-propagation}
    \frac{1}{2} \bigl[ \I \otimes \I - \ket{1}\bra{0} \otimes U -
      \ket{0}\bra{1} \otimes U^\ast \bigr],
  \end{equation}
  for $U$ being either $\Lambda(P)$ or $H\otimes H$.

  In the case that $U = \Lambda(P)$, one has that the projection
  \eqref{eq:projection-from-propagation} is the sum of the four
  Clifford-Hamiltonian projections corresponding to these vectors:
  \begin{equation}
    \begin{split}
      \ket{-}\ket{00} & = (ZH\otimes\I\otimes\I)\ket{000},\\
      \ket{-}\ket{01} & = (ZH\otimes\I\otimes X)\ket{000},\\
      \ket{-}\ket{10} & = (ZH\otimes X\otimes\I)\ket{000},\\
      \ket{\circlearrowright}\ket{11} & = (P^{\ast} H \otimes X \otimes
      X)\ket{000},\\
    \end{split}
  \end{equation}
  where $\ket{\circlearrowright} = (\ket{0}- i\ket{1})/\sqrt{2}$.
  In the case that $U = H\otimes H$, one has that the projection
  \eqref{eq:projection-from-propagation} is the sum of the four
  Clifford-Hamiltonian projections corresponding to these vectors:
  \begin{equation}
    \begin{split}
      \ket{\psi_1} & =
      \bigl( \ket{000} - \ket{011} - \ket{101} - \ket{110}\bigr)/2,\\
      \ket{\psi_2} & =
      \bigl( \ket{000} + \ket{011} - \ket{100} - \ket{111}\bigr)/2,\\
      \ket{\psi_3} & =
      \bigl( \ket{001} - \ket{010} + \ket{101} - \ket{110}\bigr)/2,\\
      \ket{\psi_4} & =
      \bigl( \ket{001} + \ket{010} - \ket{100} + \ket{111}\bigr)/2.
    \end{split}
  \end{equation}
  All four of these vectors are obtained by a Clifford operation applied to
  the all-zero state.
  In particular, when the following Clifford circuits are applied to the state
  $\ket{000}$, the states $\ket{\psi_1}$, $\ket{\psi_2}$, $\ket{\psi_3}$, and
  $\ket{\psi_4}$ are obtained:
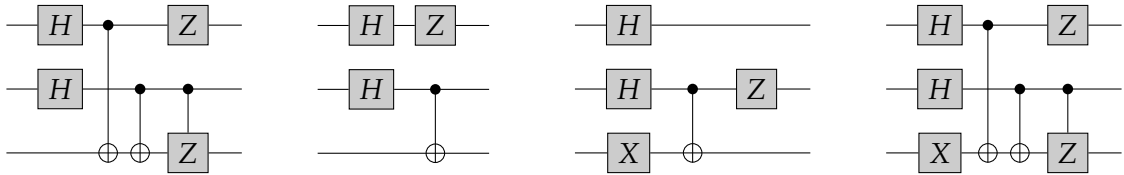
\vspace{2mm}
  \begin{center}
    \begin{tikzpicture}[scale=0.85]
      \node (In1) at (-2,1) {};
      \node (In2) at (-2,0) {};
      \node (In3) at (-2,-1) {};
      \node (Out1) at (2,1) {};
      \node (Out2) at (2,0) {};
      \node (Out3) at (2,-1) {};
      \node (H1) at (-1,1) [draw, fill=ChannelColor] {$H$};
      \node (H2) at (-1,0) [draw, fill=ChannelColor] {$H$};
      \node[circle, fill, minimum size = 4pt, inner sep=0mm]
      (Control1) at (-0.25,1) {};
      \node[circle, draw, minimum size = 7pt, inner sep=0mm]
      (Target1) at (-0.25,-1) {};
      \draw (Control1.center) -- (Target1.south);
      \node[circle, fill, minimum size = 4pt, inner sep=0mm]
      (Control2) at (0.25,0) {};
      \node[circle, draw, minimum size = 7pt, inner sep=0mm]
      (Target2) at (0.25,-1) {};
      \draw (Control2.center) -- (Target2.south);
      \node (CZ) at (1,-1) [draw, fill=ChannelColor] {$Z$};
      \node (Z) at (1,1) [draw, fill=ChannelColor] {$Z$};
      \draw (In1) -- (H1) -- (Z) -- (Out1);
      \draw (In2) -- (H2) -- (Out2);
      \draw (In3) -- (CZ) -- (Out3);
      \node[circle, fill, minimum size = 4pt, inner sep=0mm]
      (Control3) at (1,0) {};
      \draw (Control3.center) -- (CZ.north);
    \end{tikzpicture}
    \quad
    \begin{tikzpicture}[scale=0.85]
      \node (In1) at (-2,1) {};
      \node (In2) at (-2,0) {};
      \node (In3) at (-2,-1) {};
      \node (Out1) at (1,1) {};
      \node (Out2) at (1,0) {};
      \node (Out3) at (1,-1) {};
      \node (H1) at (-1,1) [draw, fill=ChannelColor] {$H$};
      \node (H2) at (-1,0) [draw, fill=ChannelColor] {$H$};
      \node (Z) at (0,1) [draw, fill=ChannelColor] {$Z$};
      \draw (In1) -- (H1) -- (Z) -- (Out1);
      \draw (In2) -- (H2) -- (Out2);
      \draw (In3) -- (Out3);
      \node (Phantom) at (1,-1) {\phantom{$Y$}};
      \node[circle, fill, minimum size = 4pt, inner sep=0mm]
      (Control) at (0,0) {};
      \node[circle, draw, minimum size = 7pt, inner sep=0mm]
      (Target) at (-0,-1) {};
      \draw (Control.center) -- (Target.south);
    \end{tikzpicture}
    \quad
    \begin{tikzpicture}[scale=0.85]
      \node (In1) at (-2,1) {};
      \node (In2) at (-2,0) {};
      \node (In3) at (-2,-1) {};
      \node (Out1) at (2,1) {};
      \node (Out2) at (2,0) {};
      \node (Out3) at (2,-1) {};
      \node (H1) at (-1,1) [draw, fill=ChannelColor] {$H$};
      \node (H2) at (-1,0) [draw, fill=ChannelColor] {$H$};
      \node (X) at (-1,-1) [draw, fill=ChannelColor] {$X$};
      \node (Z) at (1,0) [draw, fill=ChannelColor] {$Z$};
      \draw (In1) -- (H1) -- (Out1);
      \draw (In2) -- (H2) -- (Z) -- (Out2);
      \draw (In3) -- (X) -- (Out3);
      \node[circle, fill, minimum size = 4pt, inner sep=0mm]
      (Control) at (0,0) {};
      \node[circle, draw, minimum size = 7pt, inner sep=0mm]
      (Target) at (-0,-1) {};
      \draw (Control.center) -- (Target.south);
    \end{tikzpicture}
    \quad
    \begin{tikzpicture}[scale=0.85]
      \node (In1) at (-2,1) {};
      \node (In2) at (-2,0) {};
      \node (In3) at (-2,-1) {};
      \node (Out1) at (2,1) {};
      \node (Out2) at (2,0) {};
      \node (Out3) at (2,-1) {};
      \node (H1) at (-1,1) [draw, fill=ChannelColor] {$H$};
      \node (H2) at (-1,0) [draw, fill=ChannelColor] {$H$};
      \node (Z) at (1,-1) [draw, fill=ChannelColor] {$Z$};
      \node (Z2) at (1,1) [draw, fill=ChannelColor] {$Z$};
      \node (X) at (-1,-1) [draw, fill=ChannelColor] {$X$};
      \draw (In1) -- (H1) -- (Z2) -- (Out1);
      \draw (In2) -- (H2) -- (Out2);
      \draw (In3) -- (X) -- (Z) -- (Out3);
      \node[circle, fill, minimum size = 4pt, inner sep=0mm]
      (Control1) at (-0.25,1) {};
      \node[circle, draw, minimum size = 7pt, inner sep=0mm]
      (Target1) at (-0.25,-1) {};
      \draw (Control1.center) -- (Target1.south);
      \node[circle, fill, minimum size = 4pt, inner sep=0mm]
      (Control2) at (0.25,0) {};
      \node[circle, draw, minimum size = 7pt, inner sep=0mm]
      (Target2) at (0.25,-1) {};
      \draw (Control2.center) -- (Target2.south);
      \node[circle, fill, minimum size = 4pt, inner sep=0mm]
      (Control3) at (1,0) {};
      \draw (Control3.center) -- (Z.north);
    \end{tikzpicture}
  \end{center}
  This completes the proof.
\end{proof}

\begin{remark}
  If one is given a witness to a given \class{QMA} problem $A$, it is possible
  to efficiently compute a witness to the corresponding $k$-local Hamiltonian
  problem instance through Kitaev's reduction.
  Our reduction also inherits this property.
\end{remark}

\begin{remark}
  There is no loss of generality in setting $q = 1$ in the statement of the
  $k$-LCH problem, meaning that Theorem~\ref{thm:lch} holds for this somewhat
  simplified problem statement.
  This may be proved by repeating each Hamiltonian term $q$ times in a given
  problem instance and adjusting $p$ as necessary.
\end{remark}

\begin{remark}
  States of the form $C\ket{0^k}$, for a Clifford operation $C$, are stabilizer
  states of $k$ qubits.
  Theorem~\ref{thm:lch} therefore implies that there exists a \class{QMA}
  verification procedure in which the verifier randomly chooses a $k$-qubit
  stabilizer state and checks whether the quantum witness state is orthogonal
  to it.
\end{remark}

\begin{remark}
  If one takes $U = H$ in \eqref{eq:projection-from-propagation}, the resulting
  projection operator projects onto the two-dimensional subspace spanned by the
  vectors $\ket{-}\ket{\gamma_0}$ and $\ket{+}\ket{\gamma_1}$, where
  \begin{equation}
    \label{eq:H-eigenvectors}
    \ket{\gamma_0} = \cos(\pi/8)\ket{0} + \sin(\pi/8)\ket{1}
    \quad\text{and}\quad
    \ket{\gamma_1} = \sin(\pi/8)\ket{0} - \cos(\pi/8)\ket{1}
  \end{equation}
  are eigenvectors of $H$.
  This projection cannot be expressed as a sum of Clifford-Hamiltonian
  projections, which explains why we needed to replace $H$ with $H\otimes H$ in
  the proof above.

  While considering this projection is not useful for proving
  Theorem~\ref{thm:lch}, we do obtain from it a different result.
  In particular, we obtain an alternative proof of a result due to
  Morimae, Nagaj, and Schuch~\cite{MNS16} establishing that single-qubit
  measurements and classical post-processing are sufficient for
  \class{QMA} verification.
  Reference~\cite{MNS16}  actually provides two proofs of this fact, one based
  on measurement-based quantum computation and the other based on a
  local-Hamiltonian problem type of approach similar to what we propose.
  While their local-Hamiltonian approach does not work for one-sided error
  (or $\class{QMA}_1$) verifications, ours does (as does their measurement-based
  quantum computation proof).
\end{remark}

%------------------------------------------------------------------------------%
\section{Description of the proof system}
\label{sec:proof-system-description}
%------------------------------------------------------------------------------%

In this section we describe our zero-knowledge proof system for the
local Clifford-Hamiltonian problem.
The main steps of the proof system are described in the subsections that follow,
and the entire proof system is summarized in
Figure~\ref{figure:proof-system-summary}.
Properties of the proof system, including completeness, soundness, and the
zero-knowledge property, are discussed in later sections of the paper.

As suggested previously, our proof system makes use of a bit commitment scheme,
and in the interest of simplicity in explaining and analyzing the proof system
we shall assume that this scheme is non-interactive.
One could, however, replace this non-interactive commitment scheme by a
different scheme (such as Naor's scheme with a 1-round commitment
phase~\cite{Nao91}).
Throughout this section it is to be assumed that an instance of the $k$-local
Clifford-Hamiltonian problem has been selected.
The instance describes Clifford-Hamiltonian projections
$H_1,\ldots,H_m$, each given by $H_j = C_j^{\ast}\ket{0^k}\bra{0^k} C_j$ for
$k$-qubit Clifford operations $C_1,\ldots,C_m$, along with a specification of
which of the~$n$ qubits these projections act upon.
The proof system does not refer to the parameters $p$ and $q$ in the
description of the $k$-local Clifford Hamiltonian problem, as these parameters
are only relevant to the performance of the proof system and not its
implementation.
It must be assumed, however, that the completeness parameter $2^{-p}$ is
a negligible function of the entire problem instance size in order for the
proof system to be zero-knowledge, and we will make this assumption hereafter.

\begin{figure}
  \begin{mdframed}[style=figstyle,innerleftmargin=10pt,innerrightmargin=10pt]

  \begin{trivlist}
  \item
    \emph{Prover's encoding step:}\vspace{1mm}

    The prover selects a tuple $(t,\pi,a,b)$ uniformly at random, where
    $t = t_1\cdots t_n$ for $t_1,\ldots,t_n\in\{0,+,\circlearrowright\}^N$,
    $\pi\in S_{2N}$, and $a = a_1\cdots a_n$ and $b = b_1\cdots b_n$ for
    $a_1,\ldots,a_n,b_1,\ldots,b_n\in\{0,1\}^{2N}$.
    The witness state contained in qubits $(\reg{X}_1,\ldots,\reg{X}_n)$
    is encoded into qubit tuples
    \begin{equation}
      \bigl(\reg{Y}^1_1,\ldots,\reg{Y}^1_{2N}\bigr),\,\ldots,\,
      \bigl(\reg{Y}^n_1,\ldots,\reg{Y}^n_{2N}\bigr)
    \end{equation}
    as described in the main text.
    These qubits are sent to the verifier, along with a commitment to the
    tuple $(\pi,a,b)$.
    \vspace{2mm}

  \item
    \emph{Coin flipping protocol:} \vspace{1mm}

    The prover and verifier engage in a coin flipping protocol, choosing
    a string $r$ of a fixed length uniformly at random.
    This random string $r$ determines a Hamiltonian term
    $H_r = C_r^{\ast}\ket{0^k}\bra{0^k}C_r$ that is to be tested.
    \vspace{2mm}

  \item
    \emph{Verifier's measurement:} \vspace{1mm}

    The verifier applies the Clifford operation $C_r$ transversally to the
    qubits
    \begin{equation}
      \bigl(\reg{Y}^{i_1}_1,\ldots,\reg{Y}^{i_1}_{2N}\bigr),\,\ldots,\,
      \bigl(\reg{Y}^{i_k}_1,\ldots,\reg{Y}^{i_k}_{2N}\bigr),
    \end{equation}
    and measures all of these qubits in the standard basis, for
    $(i_1,\ldots,i_k)$ being the indices of the qubits upon which the
    Hamiltonian term $H_r$ acts nontrivially.
    The result of this measurement is sent to the prover.

    \vspace{2mm}

  \item
    \emph{Prover's verification and response:} \vspace{1mm}

    The prover checks that the verifier's measurement results are consistent
    with the states of the trap qubits and the concatenated Steane code,
    aborting the proof system if not (causing the verifier to reject).
    In case the measurement results are consistent, the prover demonstrates
    that these measurement results are consistent with its prior commitment to
    $(\pi,a,b)$ and with the Hamiltonian term $H_r$, through a classical
    zero-knowledge proof system for the corresponding $\class{NP}$ statement
    described in the main text.
    The verifier accepts or rejects accordingly.
  \end{trivlist}
  \caption{Summary of the zero-knowledge proof system for the LCH problem}
  \label{figure:proof-system-summary}

  \end{mdframed}
\end{figure}

%------------------------------------------------------------------------------%
\subsection{Prover's witness encoding}
\label{sec:encoding}
%------------------------------------------------------------------------------%

Suppose $\reg{X} = (\reg{X}_1,\ldots,\reg{X}_n)$ is an $n$-tuple of single-qubit
registers.
These qubits are assumed to initially be in the prover's possession, and
store an $n$-qubit quantum state $\rho$ representing a possible witness for the
instance of the $k$-LCH problem under consideration.

The first step of the proof system requires the prover to encode the
state of~$\reg{X}$, using a scheme that consists of four steps.
Throughout the description of these steps it is to be assumed that $N$ is a
polynomially bounded function of the input size and is an even positive integer
power of 7.
In effect, $N$ acts as a security parameter (for the zero-knowledge property of
the proof system), and we take it to be an even power of 7 so that it may be
viewed as a number of qubits that could arise from a concatenated Steane code
allowing for a transversal application of Clifford operations, as described in
Section~\ref{sec:Steane-code} (in the appendix).
In particular, through an appropriate choice of $N$, one may guarantee that this
code has any desired polynomial lower-bound for the minimum non-zero Hamming
weight of its underlying classical code.

\begin{mylist}{\parindent}
\item[1.]
  For each $i=1,\ldots,n$, the qubit $\reg{X}_i$ is encoded into qubits
  $(\reg{Y}^i_1,\ldots,\reg{Y}^i_{N})$ by means of the concatenated Steane
  code.
  This results in the $N$-tuples
  \begin{equation}
    \label{eq:N-qubits}
    \bigl(\reg{Y}^1_1,\ldots,\reg{Y}^1_{N}\bigr),\,\ldots,\,
    \bigl(\reg{Y}^n_1,\ldots,\reg{Y}^n_{N}\bigr).
  \end{equation}
\item[2.]
  To each of the $N$-tuples in \eqref{eq:N-qubits}, the prover concatenates an
  additional $N$ \emph{trap qubits}, with each trap qubit being initialized
  to one of the single qubit pure states $\ket{0}$, $\ket{+}$, or
  $\ket{\circlearrowright}$, selected independently and uniformly at random.
  This results in qubits
  \begin{equation}
    \label{eq:2N-qubits}
    \bigl(\reg{Y}^1_1,\ldots,\reg{Y}^1_{2N}\bigr),\,\ldots,\,
    \bigl(\reg{Y}^n_1,\ldots,\reg{Y}^n_{2N}\bigr).
  \end{equation}
  The prover stores the string $t = t_1\cdots t_n$, for
  $t_1,\ldots,t_n\in\{0,+,\circlearrowright\}^N$ representing the randomly
  chosen states of the trap qubits.
\item[3.]
  A random permutation $\pi\in S_{2N}$ is selected, and the qubits in each of
  the $2N$-tuples \eqref{eq:2N-qubits} are permuted according to $\pi$.
  (Note that it is a single permutation $\pi$ that is selected and applied to
  all of the $2N$-tuples simultaneously.)

\item[4.]
  The quantum one-time pad is applied independently to each qubit in
  \eqref{eq:2N-qubits} (after they are permuted in step 3).
  That is, for $a_i,b_i\in\{0,1\}^{2N}$ chosen independently and uniformly
  at random, the unitary transformation $X^{a_i} Z^{b_i}$ is applied to
  $(\reg{Y}^i_1,\ldots,\reg{Y}^i_{2N})$, and the strings
  $a_i$ and $b_i$ are stored by the prover, for each $i=1,\ldots,n$.
\end{mylist}

\noindent
The randomness required by these encoding steps may be described by a tuple
$(t,\pi,a,b)$, where $t$ is the string representing the states
of the trap qubits described in step 2, $\pi\in S_{2N}$ is the permutation
applied in step 3, and $a = a_1\cdots a_n$ and $b = b_1\cdots b_n$ are binary
strings representing the Pauli operators applied in the one-time pad in step 4.
After performing the above encoding steps, the prover sends the resulting
qubits
\begin{equation}
  \reg{Y} = \bigl(\bigl(\reg{Y}^1_1,\ldots,\reg{Y}^1_{2N}\bigr),\,\ldots,\,
  \bigl(\reg{Y}^n_1,\ldots,\reg{Y}^n_{2N}\bigr)\bigr),
\end{equation}
along with a commitment
\begin{equation}
  z = \alg{commit}((\pi,a,b),s)
\end{equation}
to the tuple $\bigl(\pi,a,b\bigr)$, to the verifier.
Here we assume that $s$ is a random string chosen by the prover that allows for
this commitment.
(It is not necessary for the prover to commit to the selection of the trap
qubit states indicated by $t$, although it would not affect the properties of
the proof system if it were modified so that the prover also committed to the
trap qubit state selections.)

%------------------------------------------------------------------------------%
\subsection{Verifier's random challenge}
\label{sec:challenge}
%------------------------------------------------------------------------------%

Upon receiving the prover's encoded witness and commitment, the verifier issues
a challenge:
for a randomly selected index $j\in\{1,\ldots,m\}$, the verifier will check
that the $j$-th Hamiltonian term
\begin{equation}
  \label{eq:Clifford-Hamiltonian-term}
  H_j = C_j^{\ast} \ket{0^k}\bra{0^k} C_j
\end{equation}
is not violated.
Generally speaking, the verifier's actions in issuing this challenge are as
follows: for a certain collection of qubits, the verifier applies the Clifford
operation $C_j$ transversally to those qubits, performs a measurement with
respect to the standard basis, sends the outcomes to the prover, and then
expects the prover to demonstrate that the obtained outcomes are valid (in the
sense to be described later).

The randomly selected Hamiltonian term is to be determined by a binary
string~$r$, of a fixed length $\lceil \log m\rceil$, that should be viewed as
being chosen uniformly at random.
(In a moment we will discuss the random choice of $r$, which will be
given by the output of a coin flipping protocol that happens to be uniform
for honest participants.)
It is not important exactly how the binary strings of length
$\lceil \log m\rceil$ are mapped to the indices $\{1,\ldots,m\}$, so long as
every index is represented by at least one string---so that for a uniformly
chosen string $r$, each Hamiltonian term~$j$ is selected with a nonnegligible
probability.
We will write $H_r$ and $C_r$ in place of $H_j$ and $C_j$, and refer to the
Hamiltonian term determined by $r$, when it is convenient to do this.

It would be natural to allow the verifier to randomly determine which
Hamiltonian term is to be tested---but, as suggested above, we will assume that
the challenge is determined through a \emph{coin flipping protocol}
rather than leaving the choice to the verifier.
More specifically, throughout the present subsection, it should be assumed that
the random choice of the string $r$ that determines which challenge is issued is
the result of independent iterations of a commitment-based coin-flipping
protocol (i.e., the honest prover commits to a random $y_i\in\{0,1\}$, the
honest verifier selects $z_i\in\{0,1\}$ at random, the prover reveals $y_i$, and
the two participants agree that the $i$-th random bit of $r$ is
$r_i = y_i\oplus z_i$).
This guarantees (assuming the security of the commitment protocol) that the
choices are truly random, and greatly simplifies the analysis of the
zero-knowledge property of the proof system.
The use of such a protocol might not actually be necessary for the security of
the proof system, but we leave the investigation of whether it is necessary to
future work.

Now, let $(i_1,\ldots,i_k)$ denote the indices of the qubits upon which the
Hamiltonian term determined by the random string $r$ acts nontrivially.
The verifier applies the Clifford operation $C_r$ independently to each of the
$k$-qubit tuples
\begin{equation}
  \bigl(\reg{Y}^{i_1}_1,\ldots,\reg{Y}^{i_k}_{1}\bigr),
  \ldots,
  \bigl(\reg{Y}^{i_1}_{2N},\ldots,\reg{Y}^{i_k}_{2N}\bigr),
\end{equation}
which is equivalent to saying that $C_r$ is applied transversally to the tuples
\begin{equation}
  \label{eq:k-tuples-of-qubits}
  \bigl(\reg{Y}^{i_1}_1,\ldots,\reg{Y}^{i_1}_{2N}\bigr),\,\ldots,\,
  \bigl(\reg{Y}^{i_k}_1,\ldots,\reg{Y}^{i_k}_{2N}\bigr)
\end{equation}
that encode the qubits on which the Hamiltonian term $H_r$ acts nontrivially.
The qubits \eqref{eq:k-tuples-of-qubits} are then measured with respect to the
standard basis, and the results are sent to the prover.
We will let
\begin{equation}
  u_{i_1},\ldots,u_{i_k} \in \{0,1\}^{2N}
\end{equation}
denote the binary strings representing the verifier's standard basis
measurement outcomes (or claimed outcomes) corresponding to the measurements
of the tuples \eqref{eq:k-tuples-of-qubits}.

%------------------------------------------------------------------------------%
\subsection{Prover's check and response}
\label{sec:prover-check}
%------------------------------------------------------------------------------%

Upon receiving the verifier's claimed measurement outcomes corresponding to the
randomly selected Hamiltonian term, the prover first checks to see that these
outcomes could indeed have come from the measurements specified above, and then
tries to convince the verifier that these measurement outcomes are consistent
with the selected term.

In more detail, suppose that the Hamiltonian term determined by $r$ has
been challenged.
As above, we assume that this term acts nontrivially on the $k$ qubits indexed
by the $k$-tuple $(i_1,\ldots,i_k)$, and we will write
\begin{equation}
  u = u_{i_1}\cdots u_{i_k} \in \{0,1\}^{2kN}
\end{equation}
to denote the verifier's claimed standard basis measurement outcomes.

To define the prover's check for this string, it will be helpful to
first define a predicate $R_r$, which is a function of $t$, $\pi$, and $u$,
and essentially represents the prover's check \emph{after} it has made an
adjustment to the verifier's response to account for the one-time pad.
For each $i\in \{i_1,\ldots,i_k\}$, define strings $y_i,z_i\in\{0,1\}^N$ so that
\begin{equation}
  \pi(y_i z_i) = u_i.
\end{equation}
The predicate $R_r$ takes the value 1 if and only if these two conditions are
met:
\begin{mylist}{\parindent}
\item[1.]
  $y_i\in\D_N$ for every $i\in\{i_1,\ldots,i_k\}$, and $y_i\in\D_N^1$ for at
  least one index $i\in\{i_1,\ldots,i_k\}$.
\item[2.]
  $\bigl\langle
  z_{i_1} \cdots z_{i_k}
  \,\big|\,
  C_r^{\otimes N}
  \,\big|\,
  t_{i_1} \cdots t_{i_k}
  \bigr\rangle \not= 0$.
\end{mylist}
(Here we have written $\ket{t_{i_1} \cdots t_{i_k}}$ to denote the pure state
of $kN$ qubits obtained by tensoring the states $\ket{0}$, $\ket{+}$, and
$\ket{\circlearrowright}$ in this most natural way.)
The first condition concerns measurement outcomes corresponding to non-trap
qubits, and reflects the condition that these measurement outcomes are proper
encodings of binary values---but not all of which encode 0.
The second condition concerns the consistency of the verifier's measurements
with the trap qubits.

Next, we will define a predicate $Q_r$, which is a function of the variables
$t$, $\pi$, $a$, $b$, and $u$, where $t$, $\pi$, and $u$ are as above and
$a,b\in\{0,1\}^{2nN}$ refer to the strings used for the one-time pad.
The predicate $Q_r$ represents the prover's actual check, in the case that the
Hamiltonian term determined by $r$ has been selected, including an adjustment
to account for the one-time pad.
Let $c_1,\ldots,c_n,d_1,\ldots,d_n\in\{0,1\}^{2N}$ be the unique strings
for which the equation
\begin{equation}
  \label{eq:Clifford-and-Pauli}
  C_r^{\otimes 2N} \bigl(X^{a_1}Z^{b_1}\otimes \cdots \otimes
  X^{a_n}Z^{b_n}\bigr)
  = \alpha \bigl(X^{c_1}Z^{d_1}\otimes\cdots\otimes X^{c_n}Z^{d_n}\bigr)
  C_r^{\otimes 2N}
\end{equation}
holds for some choice of $\alpha\in\{1,i,-1,-i\}$.
The Clifford operation $C_r$ acts trivially on those qubits indexed by strings
outside of the set $\{i_1,\ldots,i_k\}$, so it must be the case that
$c_i = a_i$ and $d_i = b_i$ for $i\not\in\{i_1,\ldots,i_k\}$, but for those
indices $i\in\{i_1,\ldots,i_k\}$ it may be the case that
$c_i \not= a_i$ and $d_i \not= b_i$.
We will also write $c = c_1\cdots c_n$ and $d = d_1\cdots d_n$ for the sake of
convenience.
Given a description of the Clifford operation $C_r$ it is possible to
efficiently compute $c$ and $d$ from $a$ and $b$.
Having defined $c$ and $d$, we may now express the predicate $Q_r$ as follows:
\begin{equation}
  Q_r(t,\pi,u,a,b) = R_r\bigl(
  t,\pi,u \oplus c_{i_1}\cdots c_{i_k}\bigr).
\end{equation}
In essence, the predicate $Q_r$ checks the validity of the verifier's claimed
measurement results by first adjusting for the one-time pad, then referring to
$R_r$.

The prover evaluates the predicate $Q_r$, and aborts the proof system if
the predicate evaluates to~0 (as this is indicative of a dishonest verifier).
Otherwise, the prover aims to convince the verifier that the measurement
outcomes $u$ are consistent with the prover's encoding, and also that they are
not in violation of the Hamiltonian term $H_r$.
It does this specifically by engaging in a classical zero-knowledge proof
system for the following \class{NP} statement:
there exists a random string $s$ and an encoding key $(t,\pi,a,b)$ such that
(i)~$\alg{commit}((\pi,a,b),s)$ matches the prover's initial commitment~$z$, and
(ii)~$Q_r(t,\pi,u,a,b) = 1$.

It will be convenient later, in the analysis of the proof system, to sometimes
view $r$ as being an input to the predicates defined above.
Specifically, we define predicates
\begin{equation}
  Q(r,t,\pi,a,b,u) = Q_r(t,\pi,a,b,u)
  \quad\text{and}\quad
  R(r,t,\pi,u) = R_r(t,\pi,u)
\end{equation}
for this purpose.

%------------------------------------------------------------------------------%
\section{Completeness and soundness of the proof system}
\label{sec:completeness-and-soundness}
%------------------------------------------------------------------------------%

It is evident that the proof system described in the previous section is
complete.
For a given instance of the local Clifford Hamiltonian problem, if the
prover and verifier both behave honestly, as suggested in the description
of the proof system, the verifier will accept with precisely the same
probability that would be obtained by randomly selecting a Hamiltonian term,
measuring the original $n$-qubit witness state against the corresponding
projection, and accepting or rejecting accordingly.
For a positive problem instance, this acceptance probability is at least
$1 - 2^{-p}$ (for every choice of a random string $r$).

Next we will consider the soundness of the proof system.
We will prove that on a negative instance of the problem, the honest verifier
must reject with nonnegligible probability.
The prover initially sends to the verifier the qubits
\begin{equation}
  \label{eq:soundness-2Ntuples}
  \bigl(\reg{Y}^1_1,\ldots,\reg{Y}^1_{2N}\bigr),\,\ldots,\,
  \bigl(\reg{Y}^n_1,\ldots,\reg{Y}^n_{2N}\bigr),
\end{equation}
along with a commitment $z = \alg{commit}((\pi,a,b),s)$ to a tuple $(\pi,a,b)$.
We have assumed that the commitment is perfectly binding, so there is a
well-defined tuple $(\pi,a,b)$ that is determined by the prover's commitment
$z$.
We may assume without loss of generality that this tuple has the proper form
(meaning that $\pi\in S_{2N}$ is a permutation and $a$ and $b$ are binary
strings of length $2nN$, as specified in the description of the proof system),
as a commitment to a string not of this form must lead to rejection with
high probability in all cases.
Let $\xi$ be the state of the qubits
\begin{equation}
  \label{eq:soundness-qubits}
  \bigl(\reg{Y}^1_1,\ldots,\reg{Y}^1_{N}\bigr),\,\ldots,\,
  \bigl(\reg{Y}^n_1,\ldots,\reg{Y}^n_{N}\bigr)
\end{equation}
that is obtained by inverting the quantum one-time pad with respect to the
strings $a$ and $b$, inverting the permutation of each of the tuples
\eqref{eq:soundness-2Ntuples} with respect to the permutation $\pi$, and
discarding the last $N$ qubits within each tuple (i.e., the trap qubits).
For an honest prover, the state $\xi$ would be the state obtained by encoding
the original witness state using the concatenated Steane code---although
in general it cannot be assumed that $\xi$ arises in this way.
Although the verifier is not capable of recovering the state $\xi$ on its own,
because it does not know $(\pi,a,b)$, it will nevertheless be helpful to refer
to the state $\xi$ for the purposes of establishing the soundness condition of
the proof system.

We will define a collection of $N$-qubit projections operators and a channel
from $N$ qubits to one that will be useful for establishing soundness.
First, let
\begin{equation}
  \Pi_0 = \sum_{x\in\D_N^0} \ket{x}\bra{x}
  \qquad\text{and}\qquad
  \Pi_1 = \sum_{x\in\D_N^1} \ket{x}\bra{x},
\end{equation}
where $\D_N^0$ and $\D_N^1$ are subsets of $\{0,1\}^N$ representing
classical code words of the concatenated Steane code.
A standard basis measurement of any qubit encoded using this code will
necessarily yield an outcome in one of these two sets:
an encoded $\ket{0}$ state yields an outcome in $\D_N^0$, and an encoded
$\ket{1}$ state yields an outcome in $\D_N^1$.
The projections $\Pi_0$ and $\Pi_1$ therefore correspond to these two
possibilities, while the projection operator $\I - (\Pi_0 + \Pi_1)$ corresponds
to the situation in which a standard basis measurement has yielded a result
outside of the classical code space $\D_N = \D_N^0 \cup \D_N^1$.
Also define projections
\begin{equation}
  \Delta_0 = \frac{\I^{\otimes N} + Z^{\otimes N}}{2}
  \qquad\text{and}\qquad
  \Delta_1 = \frac{\I^{\otimes N} - Z^{\otimes N}}{2},
\end{equation}
which are the projections onto the spaces spanned by all even- and odd-parity
standard basis states, respectively.
It holds that $\Pi_0 \leq \Delta_0$ and $\Pi_1 \leq \Delta_1$, as the codewords
in $\D_N^0$ all have even parity and the codewords in $\D_N^1$ all have odd
parity.
Finally, define a channel $\Xi_N$, mapping $N$ qubits to 1 qubit, as follows:
\begin{equation}
  \Xi_N(\sigma) = \frac{
    \ip{\I^{\otimes N}}{\sigma} \I +
    \ip{X^{\otimes N}}{\sigma} X +
    \ip{Y^{\otimes N}}{\sigma} Y +
    \ip{Z^{\otimes N}}{\sigma} Z}{2},
\end{equation}
for every $N$-qubit operator $\sigma$.
It is evident that this mapping preserves trace, and is completely positive
when $N \equiv 1\:(\bmod\:4)$, which holds because $N$ is an even power of $7$.
One may observe that the adjoint mapping to $\Xi_N$ is given by
\begin{equation}
  \Xi_N^{\ast}(\tau) = \frac{
    \ip{\I}{\tau} \I^{\otimes N} +
    \ip{X}{\tau} X^{\otimes N} +
    \ip{Y}{\tau} Y^{\otimes N} +
    \ip{Z}{\tau} Z^{\otimes N}}{2},
\end{equation}
and satisfies
\begin{equation}
  \Xi_N^{\ast}(\ket{0}\bra{0}) = \Delta_0
  \qquad\text{and}\qquad
  \Xi_N^{\ast}(\ket{1}\bra{1}) = \Delta_1.
\end{equation}

Now, consider the state
$\rho = \Xi_N^{\otimes n}(\xi)$ of the qubits $(\reg{X}_1,\ldots,\reg{X}_n)$
that is obtained from $\xi$ when $\Xi_N$ is applied independently to each of
the $N$-tuples of qubits in \eqref{eq:soundness-qubits}.
We will prove that the verifier must reject with nonnegligible probability
for a given choice of $r$ provided that $\rho$ violates the corresponding
Hamiltonian term $H_r$.
Because every $n$-qubit state creates a nonnegligible violation in at least one
Hamiltonian term for a negative problem instance, this will suffice to prove the
soundness of the proof system.

For each random string $r$ generated by the coin flipping procedure, one may
define a measurement on the state $\xi$ that corresponds to the verifier's
actions and final decision to accept or reject given this choice of $r$,
assuming the prover behaves optimally after the coin flipping and the
verifier's measurement take place.
Specifically, corresponding to the Hamiltonian term
$H_r = C_r^{\ast}\ket{0^k}\bra{0^k}C_r$, acceptance is represented by a
projection operator $\Lambda_r$ on the qubits
\begin{equation}
  \bigl(\reg{Y}^{i_1}_1,\ldots,\reg{Y}^{i_1}_N\bigr),\ldots,
  \bigl(\reg{Y}^{i_k}_1,\ldots,\reg{Y}^{i_k}_N\bigr)
\end{equation}
defined as follows:
\begin{equation}
  \Lambda_r =
  \sum_{\substack{z\in\{0,1\}^k\\z \not= 0^k}}
  \bigl(C_r^{\otimes N}\bigr)^{\ast}
  \bigl(\Pi_{z_1}\otimes\cdots\otimes\Pi_{z_k}\bigr)
  \bigl(C_r^{\otimes N}\bigr).
\end{equation}
The probability the verifier rejects, for a given choice of $r$, is therefore
at least $1 - \ip{\Lambda_r}{\xi}$.
Because $\Pi_0 \leq \Delta_0$ and $\Pi_1 \leq \Delta_1$, the probability of
rejection is therefore at least
\begin{equation}
  1 - \sum_{\substack{z\in\{0,1\}^k\\z \not= 0^k}}
  \ip{\bigl(C_r^{\otimes N}\bigr)^{\ast}
    \bigl(\Delta_{z_1}\otimes\cdots\otimes\Delta_{z_k}\bigr)
    \bigl(C_r^{\otimes N}\bigr)}{\xi}
  =
  \ip{\bigl(C_r^{\otimes N}\bigr)^{\ast}
    \bigl(\Delta_0\otimes\cdots\otimes\Delta_0\bigr)
    \bigl(C_r^{\otimes N}\bigr)}{\xi}.
\end{equation}
By considering properties of the channel $\Xi_N$, we conclude that the verifier
rejects with probability at least
\begin{equation}
  \begin{split}
    & \ip{\bigl(C_r^{\otimes N}\bigr)^{\ast}
      \bigl(\Xi_N^{\ast}(\ket{0}\bra{0})\otimes\cdots\otimes
      \Xi_N^{\ast}(\ket{0}\bra{0})\bigr)
      \bigl(C_r^{\otimes N}\bigr)}{\xi}\\
    = & \ip{\bigl(\Xi_N^{\otimes k}\bigr)^{\ast}
      \bigl(C_r^{\ast} \ket{0^k}\bra{0^k} C_r\bigr)}{\xi}
    = \ip{C_r^{\ast} \ket{0^k}\bra{0^k} C_r}{\Xi_N^{\otimes k}(\xi)}
    = \ip{H_r}{\rho}.
  \end{split}
\end{equation}
Here we have used the observation that
\begin{equation}
  \Xi_N^{\otimes k}\bigl(C^{\otimes N}\sigma \bigl(C^{\otimes
    N}\bigr)^{\ast}\bigr)
  = C \Xi_N^{\otimes k}(\sigma) C^{\ast}
\end{equation}
for every $k$-qubit Clifford operation $C$ and every $kN$-qubit state $\sigma$,
which may be verified directly by considering the definition of $\Xi_N$.

Intuitively speaking, the argument above shows that whatever state a malicious
prover sends in the first message, one can essentially decode that state with
respect to a highly simplified variant of the encoding scheme
(after peeling off the quantum one-time pad and discarding the trap qubits),
recovering a state that would pass the Hamiltonian energy test with at least
the same probability as the verifier's acceptance probability in our
zero-knowledge proof system.
Because this probability must be bounded away from 1 on average for any
no-instance of the problem, we obtain a soundness guarantee for the proof
system.

%------------------------------------------------------------------------------%
\section{Zero-knowledge property of the proof system}
\label{sec:zero-knowledge}
%------------------------------------------------------------------------------%

In this section we will prove that the proof system described in
Section~\ref{sec:proof-system-description} is quantum computational
zero-knowledge, assuming that the commitment scheme used in the proof system is
unconditionally binding and quantum computationally concealing.
The proof has several steps, to be presented below, but first we will
summarize the main technical goal of the proof.

Figure~\ref{fig:honest-interaction} shows a diagram of the interaction between
the honest participants in the proof system.
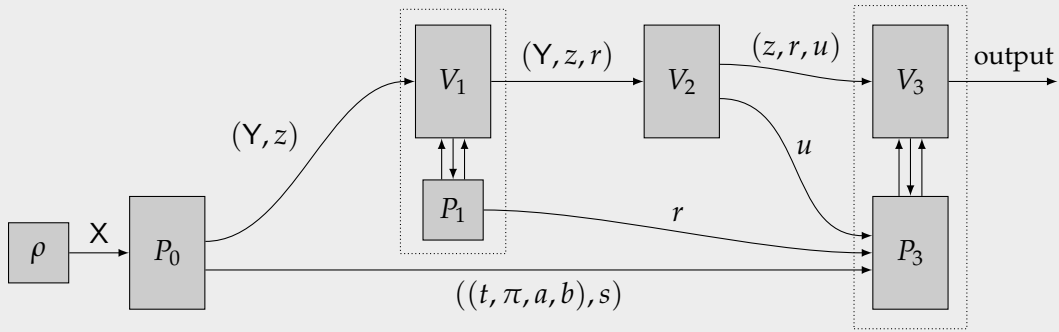
\begin{figure}[!t]
  \begin{mdframed}[style=figstyle]
  \begin{center}
    \begin{tikzpicture}[scale=0.38,
        turn/.style={draw, minimum height=15mm, minimum width=10mm,
          fill = ChannelColor, text=ChannelTextColor},
        smallturn/.style={draw, minimum height=8mm, minimum width=8mm,
          fill = ChannelColor, text=ChannelTextColor},
        >=latex]

      \node (in) at (-18.5,-3) [smallturn] {$\rho$};
      \node (P0) at (-14,-3) [turn] {$P_0$};
      \node (P1) at (-4,-1.5) [smallturn] {$P_1$};
      \node (P3) at (12,-3) [turn] {$P_3$};
      \node (V1) at (-4,3) [turn] {$V_1$};
      \node (V2) at (4,3) [turn] {$V_2$};
      \node (V3) at (12,3) [turn] {$V_3$};
      \node (out) at (17.5,3) [minimum width = 0mm] {};
      \node (phantom) at (-19.5,0) {};

      \node at (12,0) [draw, densely dotted, minimum height=43mm, minimum
          width=15mm] {};
      \node at (-4,1.25) [draw, densely dotted, minimum height=3.25cm,
        minimum width = 14mm] {};
      \node (rightphantom) at (18,3) [minimum width = 1mm] {};
      \draw[->] (in) -- (P0) node [above, midway] {$\reg{X}$};
      \draw[->] ([yshift=4mm]P0.east) .. controls +(right:30mm) and
      +(left:30mm) .. (V1.west) node [above left, pos=0.5] {$(\reg{Y},z)$};
      \draw[->] ([yshift=-6mm]P0.east) -- ([yshift=-6mm]P3.west)
      node [below, midway] {$((t,\pi,a,b),s)$};
      \draw[->] (P1.east) .. controls +(right:40mm) and +(left:40mm) ..
      (P3.west) node [above, midway] {$r$};
      \draw[->] (V1.east) -- (V2.west) node [above, midway] {$(\reg{Y},z,r)$};
      \draw[->] ([yshift=-6mm]V2.east) .. controls +(right:30mm) and
      +(left:30mm) .. ([yshift=6mm]P3.west) node [right, pos=0.4] {$u$};
      \draw[->] ([yshift=6mm]V2.east) .. controls +(right:20mm) and
      +(left:20mm) .. (V3.west) node [above, midway] {$(z,r,u)$};
      \draw[->] ([xshift=-4mm]P1.north) -- ([xshift=-4mm]V1.south);
      \draw[->] (V1.south) -- (P1.north);
      \draw[->] ([xshift=4mm]P1.north) -- ([xshift=4mm]V1.south);
      \draw[->] ([xshift=-4mm]P3.north) -- ([xshift=-4mm]V3.south);
      \draw[->] (V3.south) -- (P3.north);
      \draw[->] ([xshift=4mm]P3.north) -- ([xshift=4mm]V3.south);
      \draw[->] (V3.east) -- (out.west) node [above, pos=0.6] {\small output};
    \end{tikzpicture}
  \end{center}
  \caption{The interaction between honest participants.
    The prover's quantum witness $\rho$ is encoded into $\reg{Y}$ together
    with the encoding key $(t,\pi,a,b)$ by the prover's action $P_0$.
    The string $z$ represents the prover's commitment to $(\pi,a,b)$ and
    the string $s$ represents random bits used by the prover to implement
    this commitment.
    The string $r$ represents the random bits generated by the coin flipping
    protocol, which is depicted within the dotted rectangle on the left.
    The string $u$ represents the verifier's standard basis measurements for a
    subset of the qubits of $\reg{Y}$ determined by the challenge corresponding
    to the random string $r$.
    The classical zero-knowledge protocol is depicted within the dotted
    rectangle on the right.}
  \label{fig:honest-interaction}
  \end{mdframed}
\end{figure}
A cheating verifier aiming to extract knowledge from the prover might, of
course, not follow the prescribed actions of the honest verifier.
In particular, the cheating verifier may take a quantum register as input, store
quantum information in between its actions, and output a quantum register.
Figure~\ref{fig:cheating-verifier} illustrates such a cheating verifier
interacting with the honest prover.
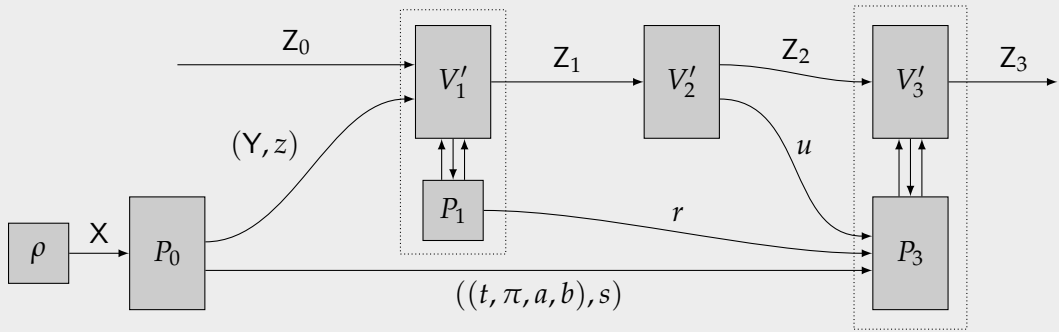
\begin{figure}[!t]
  \begin{mdframed}[style=figstyle]
  \begin{center}
    \begin{tikzpicture}[scale=0.38,
        turn/.style={draw, minimum height=15mm, minimum width=10mm,
          fill = ChannelColor, text=ChannelTextColor},
        smallturn/.style={draw, minimum height=8mm, minimum width=8mm,
          fill = ChannelColor, text=ChannelTextColor},
        >=latex]

      \node (in) at (-18.5,-3) [smallturn] {$\rho$};
      \node (P0) at (-14,-3) [turn] {$P_0$};
      \node (P1) at (-4,-1.5) [smallturn] {$P_1$};
      \node (P3) at (12,-3) [turn] {$P_3$};
      \node (V0) at (-14,3) {};
      \node (V1) at (-4,3) [turn] {$V_1'$};
      \node (V2) at (4,3) [turn] {$V_2'$};
      \node (V3) at (12,3) [turn] {$V_3'$};
      \node (out) at (17.5,3) {};
      \node (phantom) at (-19.5,0) {};

      \draw[->] ([yshift=6mm]V0.east) -- ([yshift=6mm]V1.west)
      node [above, midway] {$\reg{Z}_0$};
      \node at (12,0) [draw, densely dotted, minimum height=43mm, minimum
          width=15mm] {};
      \node at (-4,1.25) [draw, densely dotted, minimum height=3.25cm,
        minimum width = 14mm] {};
      \node (rightphantom) at (18,3) [minimum width = 1mm] {};
      \draw[->] (in) -- (P0) node [above, midway] {$\reg{X}$};
      \draw[->] ([yshift=4mm]P0.east) .. controls +(right:30mm) and
      +(left:30mm) .. ([yshift=-6mm]V1.west) node [above left, pos=0.5]
      {$(\reg{Y},z)$};
      \draw[->] ([yshift=-6mm]P0.east) -- ([yshift=-6mm]P3.west)
      node [below, midway] {$((t,\pi,a,b),s)$};
      \draw[->] (P1.east) .. controls +(right:40mm) and +(left:40mm) ..
      (P3.west) node [above, midway] {$r$};
      \draw[->] (V1.east) -- (V2.west) node [above, midway] {$\reg{Z}_1$};
      \draw[->] ([yshift=-6mm]V2.east) .. controls +(right:30mm) and
      +(left:30mm) .. ([yshift=6mm]P3.west) node [right, pos=0.4] {$u$};
      \draw[->] ([yshift=6mm]V2.east) .. controls +(right:20mm) and
      +(left:20mm) .. (V3.west) node [above, midway] {$\reg{Z}_2$};
      \draw[->] ([xshift=-4mm]P1.north) -- ([xshift=-4mm]V1.south);
      \draw[->] (V1.south) -- (P1.north);
      \draw[->] ([xshift=4mm]P1.north) -- ([xshift=4mm]V1.south);
      \draw[->] ([xshift=-4mm]P3.north) -- ([xshift=-4mm]V3.south);
      \draw[->] (V3.south) -- (P3.north);
      \draw[->] ([xshift=4mm]P3.north) -- ([xshift=4mm]V3.south);
      \draw[->] (V3.east) -- (out.west) node [above, pos=0.6] {$\reg{Z}_3$};
    \end{tikzpicture}
  \end{center}
  \caption{A potentially dishonest verifier takes an auxiliary quantum register
    $\reg{Z}_0$ as input, may store quantum information (represented by
    registers $\reg{Z}_1$ and $\reg{Z}_2$), and outputs quantum information
    stored in register $\reg{Z}_3$.}
  \label{fig:cheating-verifier}
  \end{mdframed}
\end{figure}
The goal of the proof is to demonstrate that, for any cheating verifier
of the form suggested by Figure~\ref{fig:cheating-verifier}, there exists an
efficient simulator that implements a channel from $\reg{Z}_0$ to $\reg{Z}_3$
that is computationally indistinguishable from the channel implemented by the
cheating verifier and prover interaction.
In particular, the simulator does not have access to the witness state $\rho$.

\subsubsection*{Step 1: simulating the coin flipping protocol}

By the results of~\cite{DL09}, there must exist an efficient simulator
$S_1$ for the interaction of $V_1'$ with $P_1$.
To be more precise, for $S_1$ being given an input of the same form as $V_1'$,
along with a uniformly chosen random string $r$ of the length required by our
proof system, the resulting action is quantum computationally indistinguishable
from $V_1'$ interacting with $P_1$.
Figure~\ref{fig:coin-flipping-simulated} illustrates the process that is
obtained by performing this substitution.
\begin{figure}
  \begin{mdframed}[style=figstyle]
  \begin{center}
    \begin{tikzpicture}[scale=0.38,
        turn/.style={draw, minimum height=15mm, minimum width=10mm,
          fill = ChannelColor, text=ChannelTextColor},
        smallturn/.style={draw, minimum height=8mm, minimum width=8mm,
          fill = ChannelColor, text=ChannelTextColor},
        >=latex]

      \node (in) at (-18.5,-3) [smallturn] {$\rho$};
      \node (P0) at (-14,-3) [turn] {$P_0$};
      \node (coins) at (-4,-1.75) [smallturn] {\small coins};
      \node (P3) at (12,-3) [turn] {$P_3$};
      \node (V0) at (-14,3) {};
      \node (V1) at (-4,3) [turn] {$S_1$};
      \node (V2) at (4,3) [turn] {$V_2'$};
      \node (V3) at (12,3) [turn] {$V_3'$};
      \node (out) at (17.5,3) {};
      \node (phantom) at (-19.5,0) {};

      \draw[->] ([yshift=6mm]V0.east) -- ([yshift=6mm]V1.west)
      node [above, midway] {$\reg{Z}_0$};
      \node at (12,0) [draw, densely dotted, minimum height=43mm, minimum
          width=15mm] {};
      \node (rightphantom) at (18,3) [minimum width = 1mm] {};
      \draw[->] (in) -- (P0) node [above, midway] {$\reg{X}$};
      \draw[->] ([yshift=4mm]P0.east) .. controls +(right:30mm) and
      +(left:30mm) .. ([yshift=-6mm]V1.west) node [above left, pos=0.5]
      {$(\reg{Y},z)$};
      \draw[->] ([yshift=-6mm]P0.east) -- ([yshift=-6mm]P3.west)
      node [below, midway] {$((t,\pi,a,b),s)$};
      \draw[->] (coins.east) .. controls +(right:40mm) and +(left:40mm) ..
      (P3.west) node [above, midway] {$r$};
      \draw[->] (V1.east) -- (V2.west) node [above, midway] {$\reg{Z}_1$};
      \draw[->] ([yshift=-6mm]V2.east) .. controls +(right:30mm) and
      +(left:30mm) .. ([yshift=6mm]P3.west) node [right, pos=0.4] {$u$};
      \draw[->] ([yshift=6mm]V2.east) .. controls +(right:20mm) and
      +(left:20mm) .. (V3.west) node [above, midway] {$\reg{Z}_2$};
      \draw[->] (coins.north) -- (V1.south) node [left, midway] {$r$};
      \draw[->] ([xshift=-4mm]P3.north) -- ([xshift=-4mm]V3.south);
      \draw[->] (V3.south) -- (P3.north);
      \draw[->] ([xshift=4mm]P3.north) -- ([xshift=4mm]V3.south);
      \draw[->] (V3.east) -- (out.west) node [above, pos=0.6] {$\reg{Z}_3$};
    \end{tikzpicture}
  \end{center}
  \caption{The interaction corresponding to the execution of the coin
    flipping protocol has been replaced by a simulator $S_1$ along with a
    true random string generator (labeled \emph{coins}).}
  \label{fig:coin-flipping-simulated}
  \end{mdframed}
\end{figure}
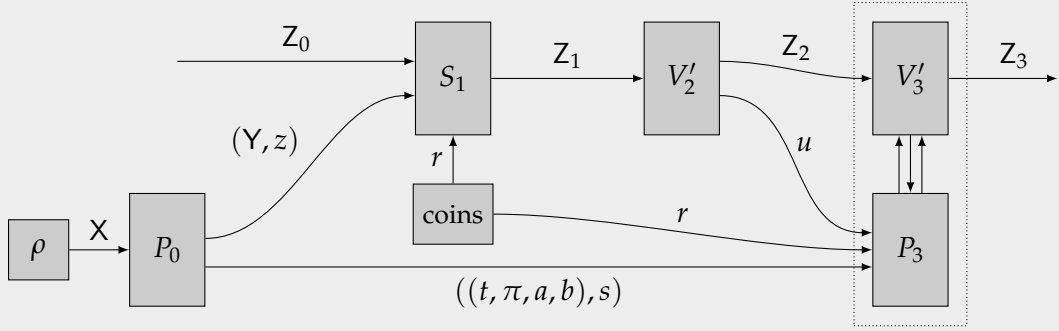
As the simulator $S_1$ together with the true random string generator is
computationally indistinguishable from the interaction between $V_1'$ and
$P_1$, the process illustrated in Figure~\ref{fig:coin-flipping-simulated}
is computationally indistinguishable from the process illustrated in
Figure~\ref{fig:cheating-verifier}.
It therefore suffices for us to prove that the process illustrated in
Figure~\ref{fig:coin-flipping-simulated} can be efficiently simulated
(without access to the witness state $\rho$).

\subsubsection*{Step 2: simulating the classical zero-knowledge protocol}

In the next step of the proof, we replace the interaction between a cheating
verifier $V_3'$ and the prover~$P_3$ in the classical zero-knowledge protocol
by an efficient simulation.

The prover holds an encoding key $(t,\pi,a,b)$ along with a random string
$s$ it has used to commit to the tuple $(\pi,a,b)$.
The commitment $z = \alg{commit}((\pi,a,b),s)$ was sent to the verifier,
together with the encoding register $\reg{Y}$, in the first step of the proof
system.
The verifier sends a string $u$ that, in the honest case, represents the output
of a measurement of some subset of the qubits of $\reg{Y}$ with respect to the
standard basis, after the transversal application of a Clifford operation
depending on the random choice of $r$.
The statement that the honest prover aims to prove in the classical
zero-knowledge protocol is that there exists an encoding key $(t,\pi,a,b)$
along with a string $s$ such that $z = \alg{commit}((\pi,a,b),s)$ and
$Q(r,t,\pi,a,b,u) = 1$.
The honest prover always holds an encoding key $(t,\pi,a,b)$ and a binary
string $s$ for which $z = \alg{commit}((\pi,a,b),s)$, and if it is the case that
$Q(r,t,\pi,a,b,u) = 0$, the honest prover aborts.
By the assumption that the classical zero-knowledge protocol is indeed
computational zero-knowledge, there must therefore exist an efficient simulator
$S_3$ so that the process described in
Figure~\ref{fig:zero-knowledge-simulated} is computationally indistinguishable
from the one described by Figure~\ref{fig:coin-flipping-simulated}.
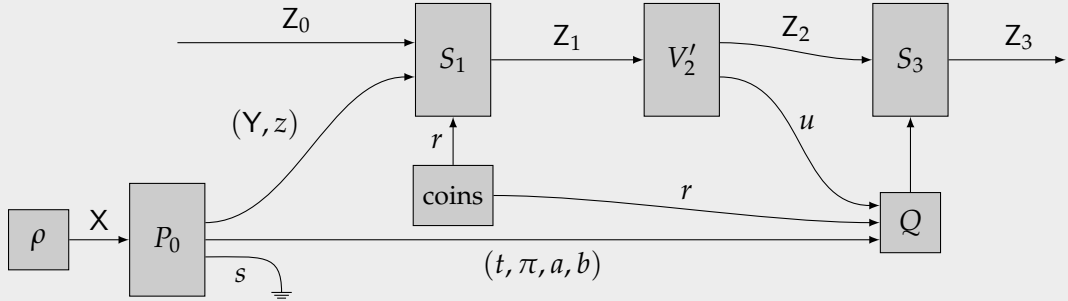
\begin{figure}[!t]
  \begin{mdframed}[style=figstyle]
  \begin{center}
    \begin{tikzpicture}[scale=0.38,
        turn/.style={draw, minimum height=15mm, minimum width=10mm,
          fill = ChannelColor, text=ChannelTextColor},
        smallturn/.style={draw, minimum height=8mm, minimum width=8mm,
          fill = ChannelColor, text=ChannelTextColor},
        >=latex]

      \node (in) at (-18.5,-3.3) [smallturn] {$\rho$};
      \node (P0) at (-14,-3.3) [turn] {$P_0$};
      \node (coins) at (-4,-1.75) [smallturn] {\small coins};
      \node (Q) at (12,-2.7) [smallturn] {$Q$};
      \node (V0) at (-14,3) {};
      \node (V1) at (-4,3) [turn] {$S_1$};
      \node (V2) at (4,3) [turn] {$V_2'$};
      \node (V3) at (12,3) [turn] {$S_3$};
      \node (out) at (17.5,3) [inner sep = 0mm] {};
      \node (nowhere) at (-10,-5.5) {};
      \node (phantom) at (-19.5,0) {};

      \draw[->] ([yshift=6mm]V0.east) -- ([yshift=6mm]V1.west)
      node [above, midway] {$\reg{Z}_0$};
      \node (rightphantom) at (18,3) [minimum width = 1mm] {};
      \draw[->] (in) -- (P0) node [above, midway] {$\reg{X}$};
      \draw[->] ([yshift=6mm]P0.east) .. controls +(right:30mm) and
      +(left:30mm) .. ([yshift=-6mm]V1.west) node [above left, pos=0.5]
      {$(\reg{Y},z)$};
      \draw[->] (P0.east) -- ([yshift=-6mm]Q.west)
      node [below, midway] {$(t,\pi,a,b)$};
      \draw[->] (coins.east) .. controls +(right:40mm) and +(left:40mm) ..
      (Q.west) node [above, midway] {$r$};
      \draw[->] (V1.east) -- (V2.west) node [above, midway] {$\reg{Z}_1$};
      \draw[->] ([yshift=-6mm]V2.east) .. controls +(right:30mm) and
      +(left:30mm) .. ([yshift=6mm]Q.west) node [right, pos=0.4] {$u$};
      \draw[->] ([yshift=6mm]V2.east) .. controls +(right:20mm) and
      +(left:20mm) .. (V3.west) node [above, midway] {$\reg{Z}_2$};
      \draw[->] (coins.north) -- (V1.south) node [left, midway] {$r$};
      \draw[->] (Q.north) -- (V3.south);
      \draw[->] (V3.east) -- (out.west) node [above, pos=0.6] {$\reg{Z}_3$};

      \draw ([yshift=-6mm]P0.east) .. controls +(right:20mm) and
      +(up:14mm) .. (nowhere.north) node [below left, pos=0.4] {$s$};

      % Drawing a cheesy little ground symbol...
      \draw (nowhere.north west) -- (nowhere.north east);
      \draw ([xshift=1mm, yshift=-1mm]nowhere.north west) --
      ([xshift=-1mm, yshift=-1mm]nowhere.north east);
      \draw ([xshift=2mm, yshift=-2mm]nowhere.north west) --
      ([xshift=-2mm, yshift=-2mm]nowhere.north east);

    \end{tikzpicture}
  \end{center}
  \caption{The interaction corresponding to the execution of the classical
    zero-knowledge protocol has been replaced by a simulator $S_3$ along with
    the predicate~$Q$.
    It is assumed that when the output of $Q$ is 0, the simulator $S_3$ behaves
    as the cheating verifier $V_3'$ would when the prover aborts the proof
    system.
    The string $s$ produced by $P_0$ in forming the commitment to $(\pi,a,b)$
    is discarded.}
  \label{fig:zero-knowledge-simulated}
  \end{mdframed}
\end{figure}
Note that the string $s$ used by $P_0$ to form the commitment
$z = \alg{commit}((\pi,a,b),s)$ can be discarded immediately after $P_0$ is
run.

\subsubsection*{Step 3: eliminating the commitment}

\begin{figure}[!t]
  \begin{mdframed}[style=figstyle]
  \begin{center}
    \begin{tikzpicture}[scale=0.38,
        turn/.style={draw, minimum height=15mm, minimum width=10mm,
          fill = ChannelColor, text=ChannelTextColor},
        smallturn/.style={draw, minimum height=8mm, minimum width=8mm,
          fill = ChannelColor, text=ChannelTextColor},
        >=latex]

      \node at (-3,1.2) [draw, densely dotted, minimum width = 9cm,
        minimum height = 3.5cm, fill=black!10] {};

      \node (in) at (-18.5,-3.3) [minimum width = 1mm] {};
      \node (commit) at (-10,-1) [smallturn] {
        \small\begin{tabular}{@{}c@{}}
          commit\\ $(\pi_0,a_0,b_0)$
        \end{tabular}
      };
      \node (V0) at (-21,3) {};
      \node (P0) at (-21,-2) {};
      \node (S1) at (-4,3) [turn] {$S_1$};
      \node (V2) at (4,3) [turn] {$V_2'$};
      \node (V3) at (12,3) {};
      \node (coins) at (-4,-6) {};
      \node (Q) at (15,-2) {};
      \node (S3) at (15,3) {};
      \node at (7.75,-2.5) {$V'$};

      \draw[->] ([yshift=8mm]V0.east) -- ([yshift=8mm]S1.west)
      node [above, pos = 0.1] {$\reg{Z}_0$};

      \draw[->] (commit.north) .. controls +(up:20mm) and
      +(left:30mm) .. ([yshift=-8mm]S1.west) node [below, pos=0.8] {$z$};

      \draw[->] (S1.east) -- (V2.west) node [above, midway] {$\reg{Z}_1$};

      \draw[->] (P0.east) .. controls +(right:60mm) and
      +(left:130mm) .. (S1.west) node [above left, pos=0.15] {$\reg{Y}$};

      \draw[->] ([yshift=-6mm]V2.east) .. controls +(right:30mm) and
      +(left:30mm) .. (Q.west) node [above, pos=0.8] {$u$};

      \draw[->] (coins) -- (S1) node [right,pos=0.15] {$r$};

      \draw[->] ([yshift=6mm]V2.east) -- ([yshift=6mm]S3.west) node [above,
        pos = 0.8] {$\reg{Z}_2$};

    \end{tikzpicture}
  \end{center}
  \caption{The commitment to a fixed tuple $(\pi_0,a_0,b_0)$, the simulator
    $S_1$, and the dishonest verifier action $V_2'$ may be merged into a single
    efficiently implementable action $V'$ that represents an attack against
    the encoding scheme.}
  \label{fig:V'}
  \end{mdframed}
\end{figure}
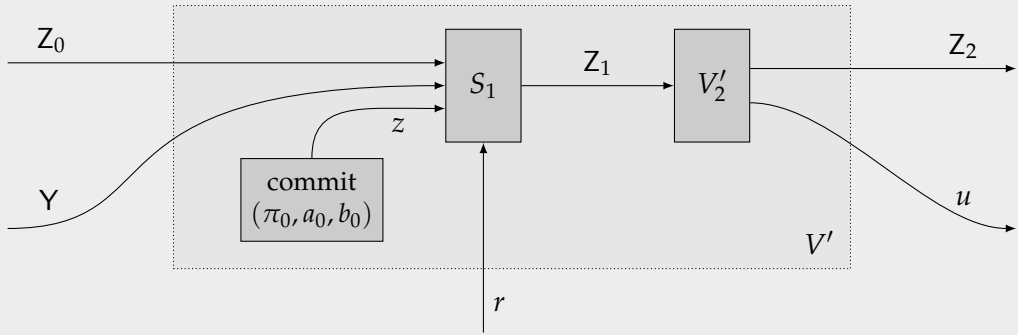
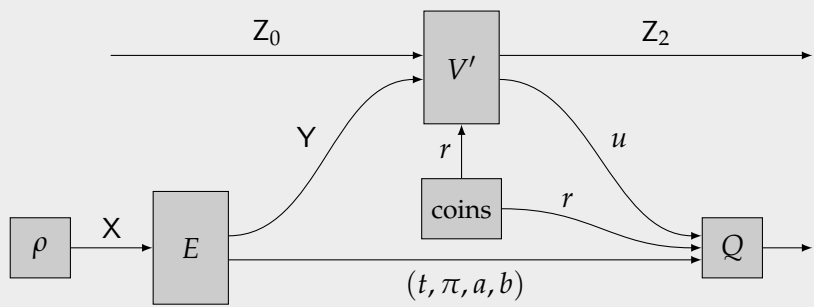
\begin{figure}[t]
  \begin{mdframed}[style=figstyle]
  \begin{center}
    \begin{tikzpicture}[scale=0.40,
        turn/.style={draw, minimum height=15mm, minimum width=10mm,
          fill = ChannelColor, text=ChannelTextColor},
        smallturn/.style={draw, minimum height=8mm, minimum width=8mm,
          fill = ChannelColor, text=ChannelTextColor},
        protocol/.style={draw, densely dotted, minimum height=43mm, minimum
          width=14mm},
        measure/.style={draw, minimum width=7mm, minimum height=7mm,
          fill = ChannelColor},
        >=latex]

      \node (V0) at (-16,3) {};
      \node (V4) at (10,3) {};
      \node (P0) at (-13,-3) [turn] {$E$};
      \node (V1) at (-4,3) [turn] {$V'$};
      \node (P3) at (5,-3) [smallturn] {$Q$};
      \node (V3) at (8,3) {};
      \node (P4) at (8,-3) {};
      \node (coins) at (-4,-1.7) [smallturn] {\small coins};
      \node (rho) at (-18,-3) [smallturn] {$\rho$};
      \node (phantom) at (-19.5,0) {};

      \draw[->] (coins.north) -- (V1.south) node [left, midway] {$r$};

      \draw[->] (coins.east) .. controls +(right:30mm) and
      +(left:30mm) .. (P3.west) node [above, pos=0.3] {$r$};

      \draw[->] ([yshift=4mm]P0.east) .. controls +(right:30mm) and
      +(left:30mm) .. ([yshift=-4mm]V1.west) node [above left, pos=0.5]
      {$\reg{Y}$};

      \draw[->] ([yshift=4mm]V0.east) -- ([yshift=4mm]V1.west)
      node [above, midway] {$\reg{Z}_0$};

      \draw[->] ([yshift=-4mm]P0.east) -- ([yshift=-4mm]P3.west)
      node [below, midway] {$(t,\pi,a,b)$};

      \draw[->] ([yshift=-4mm]V1.east) .. controls +(right:30mm) and
      +(left:30mm) .. ([yshift=4mm]P3.west) node [above right, pos=0.5] {$u$};

      \draw[->] ([yshift=4mm]V1.east) -- ([yshift=4mm]V3.west)
      node [above, midway] {$\reg{Z}_2$};

      \draw[->] (P3.east) -- (P4.west);

      \draw[->] (rho)--(P0) node [above, midway] {$\reg{X}$};

    \end{tikzpicture}
  \end{center}
  \caption{A cheating verifier $V'$ aims to extract knowledge from the
    encoding of a register $\reg{X}$.}
  \label{fig:encoding-attacker}
\end{mdframed}
\end{figure}

The next step is to eliminate the commitment.
Because it is assumed that the commitment scheme is quantum computationally
concealing, and the commitment is never revealed by the process described in
Figure~\ref{fig:zero-knowledge-simulated}, this process is computationally
indistinguishable from a similar process in which the commitment
$z$ is made to a \emph{fixed} choice of a tuple $(\pi_0,a_0,b_0)$, independent
of the prover's encoding key.
In particular, one may take $\pi_0$ to be the identity permutation and
$a_0$ and $b_0$ to be all-zero strings of length $2nN$.
One may now consider the commitment to this fixed tuple $(\pi_0,a_0,b_0)$,
together with the simulator $S_1$ and the cheating verifier action $V_2'$, to
form a single, efficiently implementable action $V'$ as suggested by
Figure~\ref{fig:V'}.

The interaction between this new action $V'$ and the prover's encoding, the
random string generator, and the predicate $Q$, as is illustrated in
Figure~\ref{fig:encoding-attacker}, may now be considered.
If it is proved that the channel implemented by this process can be efficiently
simulated, then it will follow that the channel implemented by the
process described in Figure~\ref{fig:zero-knowledge-simulated} can be
efficiently simulated (in a computationally indistinguishable sense).
This is so because the composition of the process illustrated in
Figure~\ref{fig:encoding-attacker} with the efficiently implementable simulator
$S_3$ is computationally indistinguishable from the process described in
Figure~\ref{fig:zero-knowledge-simulated}.

\subsubsection*{Step 4: simulating an attack on the encoding scheme}

It therefore suffices for us to prove that, for any efficiently
implementable action $V'$, the channel implemented by the process
described by Figure~\ref{fig:encoding-attacker} can be efficiently
simulated.  In fact, it will be possible to efficiently simulate this
channel with statistical accuracy, not just in a computationally
indistinguishable sense.  This is not surprising: we have claimed that
the computational zero-knowledge property of our proof system is based
on a computationally concealing commitment scheme, and the uses of the
commitment scheme have all been eliminated from consideration by the
steps above.

At this point we may describe the simulator directly: it is illustrated in
Figure~\ref{fig:encoding-attack-simulator}, and it represents the most
straightforward approach to obtaining a simulator.
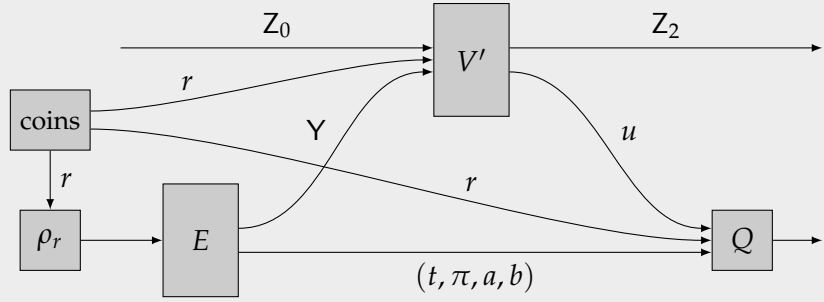
\begin{figure}[t]
  \begin{mdframed}[style=figstyle]
  \begin{center}
    \begin{tikzpicture}[scale=0.40,
        turn/.style={draw, minimum height=15mm, minimum width=10mm,
          fill = ChannelColor, text=ChannelTextColor},
        smallturn/.style={draw, minimum height=8mm, minimum width=8mm,
          fill = ChannelColor, text=ChannelTextColor},
        protocol/.style={draw, densely dotted, minimum height=43mm, minimum
          width=14mm},
        measure/.style={draw, minimum width=7mm, minimum height=7mm,
          fill = ChannelColor},
        >=latex]

      \node (V0) at (-16,3) {};
      \node (P0) at (-13,-3) [turn] {$E$};
      \node (V1) at (-4,3) [turn] {$V'$};
      \node (P3) at (5,-3) [smallturn] {$Q$};
      \node (V3) at (8,3) [minimum width = 1mm] {};
      \node (P4) at (8,-3) [minimum width = 1mm] {};
      \node (coins) at (-18,1) [smallturn] {\small coins};
      \node (rho) at (-18,-3) [smallturn] {$\rho_r$};

      \draw[->] ([yshift=3mm]coins.east) .. controls +(right:30mm) and
      +(left:30mm) .. (V1.west) node [above, pos=0.3] {$r$};
      \draw[->] ([yshift=-3mm]coins.east) .. controls +(right:50mm) and
      +(left:50mm) .. (P3.west) node [above, pos=0.6] {$r$};
      \draw[->] ([yshift=4mm]P0.east) .. controls +(right:30mm) and
      +(left:30mm) .. ([yshift=-4mm]V1.west) node [above left, pos=0.5]
      {$\reg{Y}$};
      \draw[->] ([yshift=4mm]V0.east) -- ([yshift=4mm]V1.west)
      node [above, midway] {$\reg{Z}_0$};
      \draw[->] ([yshift=-4mm]P0.east) -- ([yshift=-4mm]P3.west)
      node [below, midway] {$(t,\pi,a,b)$};
      \draw[->] ([yshift=-4mm]V1.east) .. controls +(right:30mm) and
      +(left:30mm) .. ([yshift=4mm]P3.west) node [above right, pos=0.5] {$u$};
      \draw[->] ([yshift=4mm]V1.east) -- ([yshift=4mm]V3.west)
      node [above, midway] {$\reg{Z}_2$};
      \draw[->] (P3.east) -- (P4.west);
      \draw[->] (rho) -- (P0);
      \draw[->] (coins) -- (rho) node [midway, right] {$r$};
    \end{tikzpicture}
  \end{center}
  \caption{The simulation of the process shown in
    Figure~\ref{fig:encoding-attacker} is nearly identical to that process,
    except that it uses the random string $r$ to encode a state $\rho_r$ that
    is guaranteed to pass the challenge corresponding to $r$, rather than
    encoding the witness state $\rho$.}
  \label{fig:encoding-attack-simulator}
\end{mdframed}
\end{figure}
This simulator differs from the process described in
Figure~\ref{fig:encoding-attacker} in that it uses the output of the random
string generator to choose a quantum state that, once encoded, passes the
randomly selected challenge with certainty.
It is trivial to efficiently prepare such a state given the string $r$.
It remains to prove that the channel implemented by the simulator described in
Figure~\ref{fig:encoding-attack-simulator} is indistinguishable from the
channel implemented by the process described in
Figure~\ref{fig:encoding-attacker}.
By convexity it suffices to prove that this is so for every fixed choice of the
string $r$.
\begin{figure}[!t]
  \begin{mdframed}[style=figstyle]
  \begin{center}
    \begin{tikzpicture}[scale=0.40,
        turn/.style={draw, minimum height=12mm, minimum width=10mm,
          fill = ChannelColor, text=ChannelTextColor},
        smallturn/.style={draw, minimum height=8mm, minimum width=8mm,
          fill = ChannelColor, text=ChannelTextColor},
        >=latex]

      \node (V0) at (-16,3) {};
      \node (P0) at (-13,-2) [turn] {$E$};
      \node (V1) at (-4,3) [turn] {$V'_r$};
      \node (P3) at (5,-2) [smallturn] {$Q_r$};
      \node (V3) at (8,3) [minimum width = 1mm] {};
      \node (P4) at (8,-2) [minimum width = 1mm] {};
      \node (rho) at (-17,-2) [smallturn] {$\xi$};

      \draw[->] ([yshift=4mm]P0.east) .. controls +(right:30mm) and
      +(left:30mm) .. ([yshift=-4mm]V1.west) node [above left, pos=0.5]
      {$\reg{Y}$};
      \draw[->] ([yshift=4mm]V0.east) -- ([yshift=4mm]V1.west)
      node [above, midway] {$\reg{Z}_0$};
      \draw[->] ([yshift=-4mm]P0.east) -- ([yshift=-4mm]P3.west)
      node [below, midway] {$(t,\pi,a,b)$};
      \draw[->] ([yshift=-4mm]V1.east) .. controls +(right:30mm) and
      +(left:30mm) .. ([yshift=4mm]P3.west) node [above right, pos=0.5] {$u$};
      \draw[->] ([yshift=4mm]V1.east) -- ([yshift=4mm]V3.west)
      node [above, midway] {$\reg{Z}_2$};
      \draw[->] (P3.east) -- (P4.west);
      \draw[->] (rho) -- (P0);
    \end{tikzpicture}
  \end{center}
  \caption{An arbitrary $n$-qubit state $\xi$ is encoded, and the cheating
    verifier $V'$ and predicate $Q$ for a fixed choice of a string $r$ interact
    as depicted.
    It will be proved that the channels obtained by substituting $\rho$
    and~$\rho_r$ for $\xi$ are approximately equal.}
  \label{fig:encoding-attack-xi}
\end{mdframed}
\end{figure}
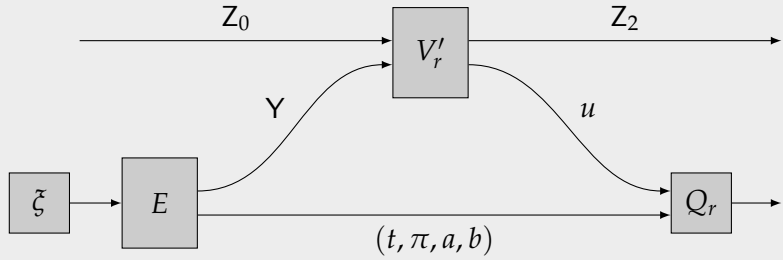

With this goal in mind, consider the process described in
Figure~\ref{fig:encoding-attack-xi}, in which an arbitrary state~$\xi$ is
encoded (corresponding either to $\rho$ or $\rho_r$ in
Figures~\ref{fig:encoding-attacker} and \ref{fig:encoding-attack-simulator}),
and the string $r$ is fixed (which has been indicated by the substitution of
$V'_r$ and $Q_r$ for $V'$ and $Q$, respectively).
We will prove that the channel implemented by any such process can have only a
limited dependence on the state~$\xi$.

More specifically, let us assume that $\xi_0$ and $\xi_1$ are arbitrary
$n$-qubit states, let $p_0$ and $p_1$ denote the probabilities with which
these two states would pass the challenge determined by $r$ (for an honest
prover and verifier pair), and let $\Psi_0$ and $\Psi_1$ denote the channels
from $\reg{Z}_0$ to $\reg{Z}_2$ together with the output bit of the predicate
$Q_r$ that are implemented by the process shown in
Figure~\ref{fig:encoding-attack-xi} when $\xi_0$ or $\xi_1$ is substituted for
$\xi$, respectively.

We claim that if the difference $\abs{p_0 - p_1}$ is negligible, then
the distance $\norm{\Psi_0-\Psi_1}_{\diamond}$ is also negligible.
The two steps that follow establish that this claim is true.
By the assumption that the prover initially holds a witness state $\rho$ that
satisfies every Hamiltonian term with probability exponentially close to 1,
this will complete the proof.

\subsubsection*{Step 5: twirling the cheating verifier}

To prove the fact suggested above regarding the channel implemented by
Figure~\ref{fig:encoding-attack-xi}, we will naturally need to make use of the
specific properties of the encoding scheme, which has not played an
important role in the analysis thus far.
The first step is to recognize that the effect of the prover's one time pad is
to \emph{twirl}\footnote{%
  The term \emph{twirl} is commonly used in quantum information theory to
  describe a process whereby a symmetrization over a collection of randomly
  chosen unitary operations has a particular effect on a state or channel.
  Twirled states and channels often take on a significantly simpler form
  than the original state or channel prior to twirling.
}
the verifier as Figure~\ref{fig:verifier-plus-otp} illustrates.
\begin{figure}[!t]
  \begin{mdframed}[style=figstyle]
  \begin{center}
    \begin{tikzpicture}[scale=0.40,
        bigturn/.style={draw, minimum height=18mm, minimum width=10mm,
          fill = ChannelColor, text=ChannelTextColor},
        smallturn/.style={draw, minimum height=8mm, minimum width=8mm,
          fill = ChannelColor, text=ChannelTextColor},
        control/.style={rounded corners = 1.5pt, draw, fill = Black,
          inner sep = 0pt, minimum size = 4pt},
        >=latex]

      \node (V) at (0,0) [bigturn] {$V'_r$};
      \node (Pauli1) at (-8,-1) [smallturn] {$X^c Z^d$};
      \node (Clifford1) at (-12,-1) [smallturn] {$C_r$};
      \node (Clifford2) at (-4,-1) [smallturn] {$C_r^{\ast}$};
      \node (Pauli2) at (4,-1) [smallturn] {$X^c$};
      \node (w) at (8,-1) [minimum width = 5mm] {$w$};
      \node (Y) at (-16,-1) [minimum width = 5mm] {$\reg{Y}$};
      \node (Z0) at (-16,1) [minimum width = 5mm] {$\reg{Z}_0$};
      \node (Z2) at (8,1) [minimum width = 5mm] {$\reg{Z}_2$};

      \draw[->] (Y) -- (Clifford1);
      \draw[->] (Clifford1) -- (Pauli1);
      \draw[->] (Pauli1) -- (Clifford2);
      \draw[->] (Clifford2) -- ([yshift=-1cm]V.west);
      \draw[->] ([yshift=-1cm]V.east) -- (Pauli2);
      \draw[->] (Pauli2) -- (w);
      \draw[->] (Z0) -- ([yshift=1cm]V.west);
      \draw[->] ([yshift=1cm]V.east) -- (Z2.west);

    \end{tikzpicture}\\[8mm]
    \begin{tikzpicture}[scale=0.40,
        turn/.style={draw, minimum height=16mm, minimum width=10mm,
          fill = ChannelColor, text=ChannelTextColor},
        smallturn/.style={draw, minimum height=8mm, minimum width=8mm,
          fill = ChannelColor, text=ChannelTextColor},
        control/.style={rounded corners = 1.5pt, draw, fill = Black,
          inner sep = 0pt, minimum size = 4pt},
        circle/.style={rounded corners = 3pt, draw,
          inner sep = 0pt, minimum size = 6pt},
        invisible/.style={minimum width=10mm},
        >=latex]

      \node (Z0) at (-8,1) [invisible] {$\reg{Z}_0$};
      \node (Y) at (-8,-2) [invisible] {$\reg{Y}$};
      \node (V) at (4,1) [smallturn] {$V''_r$};
      \node (Clifford) at (-4,-2) [smallturn] {$C_r$};
      \node (W) at (8,1) [invisible] {$\reg{Z}_2$};
      \node (Yout) at (8,-2) [invisible] {$w$};
      \node (target) at (4,-2) [circle] {};
      \node (M) at (0,-2) [smallturn] {};
      \node[draw, minimum width=5mm, minimum height=3.5mm, fill=ReadoutColor]
      (readout) at (M) {};

      \draw[thick] ($(M)+(0.3,-0.15)$) arc (0:180:3mm);
      \draw[thick] ($(M)+(0.2,0.2)$) -- ($(M)+(0,-0.2)$);
      \draw[fill] ($(M)+(0,-0.2)$) circle (0.5mm);
      \draw[->] (Clifford.east) -- (M.west);
      \draw[->] (Z0.east) -- (V.west);
      \draw[->] (Y.east) -- (Clifford.west);
      \draw[->] (V.east) -- (W.west);
      \draw[->] (M.east) -- (Yout);
      \draw (V) -- (target.south) node [right, pos = 0.45] {$v$};

    \end{tikzpicture}
  \end{center}
  \caption{The prover's one-time pad merged with the cheating verifier operation
    $V'_r$.
    Averaging over random choices of $c$ and $d$ results in a process
    that can alternatively be described as illustrated in the lower diagram.
    In this process, $V_r''$ represents a so-called \emph{quantum instrument},
    which transforms $\reg{Z}_0$ into $\reg{Z}_2$ and produces a classical
    measurement outcome.
    In this case, this classical measurement outcome is XORed onto the string
    produced by a standard basis measurement.
    (In this figure and the next, one should interpret $C_r$ and $C_r^{\ast}$
    as referring to the \emph{transversal} application of the corresponding
    Clifford operation.)}
  \label{fig:verifier-plus-otp}
\end{mdframed}
\end{figure}
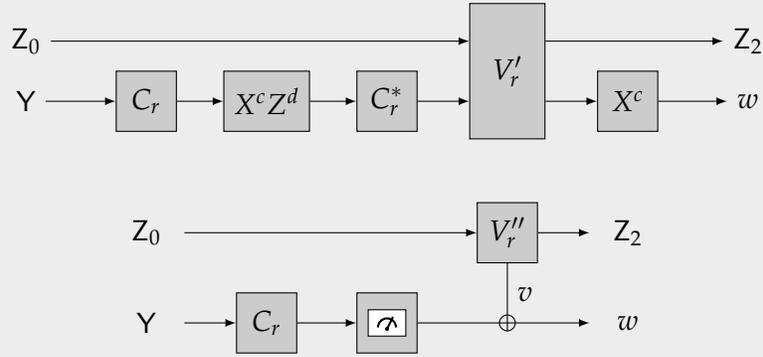

In greater detail, the last step of the encoding process is the quantum
one-time pad: the prover independently chooses one of the Pauli operations
$\I$, $X$, $Z$, or $XZ$ for each qubit of $\reg{Y}$ and applies that operation,
storing the randomly selected strings $a,b\in\Sigma^{2Nn}$.
With respect to the Clifford operation $C_r$ associated with the randomly
selected challenge (determined by the string $r$), the prover computes the pair
$(c,d)$ for which it holds that
\begin{equation}
  X^a Z^b = \bigl(C_r^{\otimes 2N}\bigr)^{\ast} X^c Z^d \bigl(C_r^{\otimes
    2N}\bigr).
\end{equation}
The first step when computing the predicate $Q_r$ is the application of $X^c$
to the string $u$, which is supposed to represent
the outcome of a standard basis measurement of a subset of the qubits after the
transversal application of $C_r$ to the corresponding qubits in the register
$\reg{Y}$.
The resulting string $w = u\oplus c$ is then fed into the predicate $R_r$
described previously.
Merging the Clifford operation $C_r^{\ast}$ with the cheating verifier
operation $V'_r$, then averaging over $c$ and $d$ chosen uniformly at random
(which is equivalent to averaging over $a$ and $b$ chosen uniformly at random),
one obtains a process of the form illustrated in the lower diagram in
Figure~\ref{fig:verifier-plus-otp}.

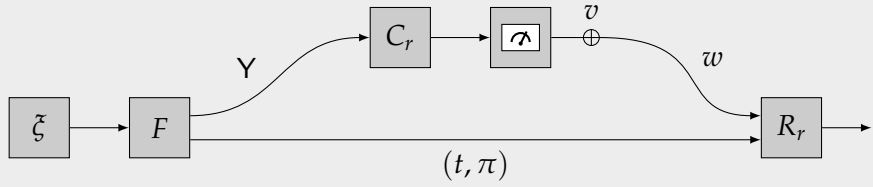
\begin{figure}[!t]
  \begin{mdframed}[style=figstyle]
  \begin{center}
    \begin{tikzpicture}[scale=0.40,
        smallturn/.style={draw, minimum height=8mm, minimum width=8mm,
          fill = ChannelColor, text=ChannelTextColor},
        measure/.style={draw, minimum width=7mm, minimum height=7mm,
          fill = ChannelColor},
        circle/.style={rounded corners = 3pt, draw,
          inner sep = 0pt, minimum size = 6pt},
        >=latex]

      \node (P0) at (-16,-3) [smallturn] {$F$};
      \node (P3) at (5,-3) [smallturn] {$R_r$};
      \node (P4) at (8,-3) [minimum width = 1mm] {};
      \node (rho) at (-20,-3) [smallturn] {$\xi$};
      \node (Clifford) at (-8,0) [smallturn] {$C_r$};
      \node (target) at (-1.65,0) [circle] {};
      \node (M) at (-4,0) [smallturn] {};
      \node[draw, minimum width=5mm, minimum height=3.5mm, fill=ReadoutColor]
      (readout) at (M) {};

      \draw[thick] ($(M)+(0.3,-0.15)$) arc (0:180:3mm);
      \draw[thick] ($(M)+(0.2,0.2)$) -- ($(M)+(0,-0.2)$);
      \draw[fill] ($(M)+(0,-0.2)$) circle (0.5mm);

      \draw[->] ([yshift=4mm]P0.east) .. controls +(right:30mm) and
      +(left:30mm) .. (Clifford.west) node [above left, pos=0.4] {$\reg{Y}$};

      \draw[->] ([yshift=-4mm]P0.east) -- ([yshift=-4mm]P3.west)
      node [below, midway] {$(t,\pi)$};

      \draw[->] (target.east) .. controls +(right:40mm) and
      +(left:30mm) .. ([yshift=4mm]P3.west) node [above right, pos=0.5] {$w$};

      \draw (M) -- (target.east);

      \draw[->] (P3.east) -- (P4.west);
      \draw[->] (rho) -- (P0);
      \draw[->] (Clifford) -- (M);

      \draw (target.north) -- (target.south);
      \node [above=0mm of target] {$v$};

    \end{tikzpicture}
  \end{center}
  \caption{An XOR attack against the prover's encoding scheme without the
    one-time pad. The transformation $F$ denotes the first three steps of the
    prover's encoding scheme.}
  \label{fig:XOR-encoding-attack}
\end{mdframed}
\end{figure}

By the observation we have just made, it suffices to consider processes of the
form described in Figure~\ref{fig:XOR-encoding-attack}, in which an $n$-qubit
state $\xi$ is encoded as described by the first three steps in the prover's
encoding procedure (but not including the one-time pad), the Clifford
operation $C_r$ (for a fixed choice of $r$) is applied transversally to the
resulting register, and the qubits on which those transversal Clifford
operations act are measured with respect to the standard basis.
For some arbitrary but fixed string $v$, the XOR of the outcome of this
measurement with $v$ is fed into the predicate $R_r$.
The process outputs a single bit, obtained by evaluating the predicate $R_r$.

\subsubsection*{Step 6: encoding security under XOR attacks}

Now let us return to the claim made previously, in which $\xi_0$ and $\xi_1$
represent $n$-qubit states, $p_0$ and $p_1$ denote the probabilities with
which these two states would pass the challenge determined by $r$ (for an
honest prover and verifier pair), and $\Psi_0$ and $\Psi_1$ denote the
channels implemented by the process shown in
Figure~\ref{fig:encoding-attack-xi} when $\xi_0$ or $\xi_1$ is substituted
for $\xi$, respectively.
If it is the case that the distribution of output bits obtained by
substituting $\xi_0$ and~$\xi_1$ for $\xi$ in
Figure~\ref{fig:XOR-encoding-attack} have negligible statistical difference,
then it follows that the difference $\norm{\Psi_0-\Psi_1}_{\diamond}$ is
also negligible.
It therefore remains to argue that the distributions obtained by
substituting $\xi_0$ and~$\xi_1$ into Figure~\ref{fig:XOR-encoding-attack}
have negligible statistical difference.

Before finishing off the last step of the analysis, it is helpful to consider
the possible outcomes of the measurement, the definition of $R_r$, and the
behavior of the procedure described in Figure~\ref{fig:XOR-encoding-attack}
when $v = 0\cdots 0$ is the all-zero string.
For any choice of $\xi$, the measurement is guaranteed to yield a string of
length $2kN$ taking the form $u_{i_1}\cdots u_{i_k}$, where
$u_{i_1},\ldots,u_{i_k}\in\{0,1\}^{2N}$ and $(i_1,\ldots,i_k)$ index the qubits
on which $C_r$ acts nontrivially.
With respect to a particular choice of $(t,\pi)$, if we define strings
$y_i,z_i\in\{0,1\}^N$, for each $i\in \{i_1,\ldots,i_k\}$, so that
\begin{equation}
  \pi(y_i z_i) = u_i,
\end{equation}
then these two conditions will necessarily be met:
\begin{mylist}{\parindent}
\item[1.]
  $y_i\in\D_N$ for every $i\in\{i_1,\ldots,i_k\}$, and
\item[2.]
  $\bigl\langle z_{i_1} \cdots z_{i_k} \,\big|\,
  C_r^{\otimes N} \,\big|\, t_{i_1} \cdots t_{i_k} \bigr\rangle \not= 0$.
\end{mylist}
Moreover, in the case that $r$ determines a Hamiltonian term challenge, the
event that $y_i\in\D_N^1$ for at least one index $i\in\{i_1,\ldots,i_k\}$ is
equivalent to $\xi$ passing this challenge.
Thus, in the case that $v = 0\cdots 0$, the process described in
Figure~\ref{fig:XOR-encoding-attack} outputs the bit 1 with precisely the
probability that an honest prover and verifier pair would result in acceptance,
assuming the prover's initial state is $\xi$ and $r$ is selected as a random
string determining the challenge.

Now let us assume that $v$ is a nonzero string, and let us consider two cases:
the first is that the Hamming weight $\abs{v}_1$ of $v$ satisfies
$\abs{v}_1 < K$, for $K$ being the minimum Hamming weight of a nonzero
codeword in $\D_N$, and the second case is that $\abs{v}_1 \geq K$.

If it is the case that $\abs{v}_1 < K$, then there are two possible ways that
the value of the predicate $R_r$ could change, in comparison to the case
$v = 0\cdots 0$.
In both cases, if there is a change, it must be from 1 to 0, caused by one of
the two conditions above becoming violated.
The first case is that one or more bits in one of the codewords
$y_{i_1},\ldots,y_{i_k}$ is flipped, causing the first condition listed
above to become violated.
The second case is that a measurement outcome for the trap qubits is obtained
that potentially violates the second condition.
Note that it is not possible that the first condition remains satisfied, but
the Hamiltonian term challenge condition that $y_i\in\D_N^1$ for at least one
index $i\in\{i_1,\ldots,i_k\}$ changes, as such a change would require at
least~$K$ bit-flips to cause a logical change in valid codewords.
It is unimportant for the purposes of the analysis to determine the probability
with which one of the two conditions becomes violated, except to observe that
it is independent of $\xi$.
(In somewhat more detail, the string $v$ may be written as
$v = v_{i_1}\cdots v_{i_k}$, and the probability that neither of the two
conditions is affected is given by the probability that $\pi^{-1}(v_{i})$
places no 1s within the first $N$ bits or over a trap qubit left in a standard
basis state within the second $N$ bits, for a random choice of $\pi$ and for
each $i\in\{i_1,\ldots,i_k\}$.)

If it is the case that $\abs{v}_1 \geq K$, then there is a possibility that,
in comparison to the functioning of the process for $v = 0\cdots 0$, the
Hamiltonian term challenge condition that $y_i\in\D_N^1$ for at least one
index $i\in\{i_1,\ldots,i_k\}$ could be affected.
That is, $v$ has enough Hamming weight to affect the logical values represented
by the codewords $y_{i_1},\ldots,y_{i_k}$.
However, as we will show, the assumption that $\abs{v}_1 \geq K$ necessarily
leads to a negligible probability that the second condition remains
satisfied---for a string $v$ having Hamming weight $K$ or higher, the
probability that none of the traps is sprung is exponentially small.
In order to argue that this is so, we require the following simple lemma.

\begin{lemma}
  Let $k$ be a positive integer, let $C$ be a Clifford operation on $k$
  qubits, and let $j\in\{1,\ldots,k\}$.
  There exists a string $t\in\{0,+,\circlearrowright\}^k$, a bit
  $a\in\{0,1\}$, and pure states $\ket{\phi_0}$ and $\ket{\phi_1}$ on
  $j-1$ qubits and $k - j$ qubits, respectively, so that
  \begin{equation}
    C \ket{t} = \ket{\phi_0} \ket{a} \ket{\phi_1}.
  \end{equation}
  Equivalently, there is a choice of $t$ so that the $j$-th qubit of
  $C\ket{t}$ is left in a standard basis state.
\end{lemma}

\begin{proof}
  The lemma is equivalent to the existence of a string $t$ so that
  $\ket{t}$ is an eigenvector of the operator
  \begin{equation}
    \label{eq:Clifford-conjugated-Pauli}
    C^{\ast}\bigl(\I^{\otimes(j-1)}\otimes Z\otimes\I^{\otimes(k-j)}\bigr)C.
  \end{equation}
  As the Clifford group normalizes the Pauli group, the operator
  \eqref{eq:Clifford-conjugated-Pauli} is a scalar multiple of a tensor product
  of Pauli operators and identity operators.
  The lemma follows from the observation that $t$ may be chosen so that each
  $\ket{t_1},\ldots,\ket{t_k}$ is an eigenvector of the Pauli operator in the
  corresponding position.
\end{proof}

By this lemma, one finds that for a random choice of
$t\in\{0,+,\circlearrowright\}^{kN}$, and for any $k$-qubit Clifford operation
$C$ applied transversally to $\ket{t}$, each qubit is left in a standard basis
state with probability at least~$3^{-k}$, and for any choice of $N$ or fewer
qubits acted on by distinct Clifford operations these events are independent.
In greater detail, if the qubits
\begin{equation}
  \bigl(\reg{Z}^{1}_1,\ldots,\reg{Z}^{k}_{1}\bigr),
  \ldots,
  \bigl(\reg{Z}^{1}_{N},\ldots,\reg{Z}^{k}_{N}\bigr)
\end{equation}
are initialized to the state $\ket{t}$, for
$t\in\{0,+,\circlearrowright\}^{kN}$ chosen uniformly at random, and
the $k$-qubit Clifford operation $C$ is applied independently to each $k$-tuple
of qubits, then each qubit is left in a standard basis state with probability
at least~$3^{-k}$, and the states of the $k$-tuples of qubits are independent.

Now we return to the analysis for a string $v$ of length $2kN$ having Hamming
weight at least~$K$.
By virtue of the fact just mentioned, it is straightforward to
obtain a negligible upper-bound on the probability for the process described in
Figure~\ref{fig:XOR-encoding-attack} to output 1.
As this event requires that a random choice of the permutation $\pi$ leaves
none of the 1-bits of $v$ in positions corresponding to trap qubits left in
standard basis states by the transversal action of $C_r$, we find that the
probability to output 1 is exponentially small in $K$.
In particular, this probability is at most
\begin{equation}
  \biggl(1 - \frac{1}{3^{k+1}}\biggr)^{K/k}
  = \exp(-\varepsilon(k) K)
\end{equation}
where $\varepsilon(k)$ denotes a positive real number depending on $k$ but not
$K$.

From a consideration of the two cases just presented, we may conclude the
following.
Suppose as before that $\xi_0$ and $\xi_1$ are $n$-qubit states that may be
substituted for $\xi$ in Figure~\ref{fig:XOR-encoding-attack}, and that
the probabilities $p_0$ and $p_1$ for these states to pass the challenge
determined by a fixed choice of $r$ have negligible difference.
Let us write $q_0(v)$ and $q_1(v)$, respectively, to denote the probability
that the process described in Figure~\ref{fig:XOR-encoding-attack} outputs 1.
As noted before, it holds that $p_0 = q_0(0\cdots 0)$ and
$p_1 = q_1(0\cdots 0)$.
For any choice of $v$ satisfying $\abs{v}_1 < K$, we have that
$q_0(v) = \beta(v) q_0(0\cdots 0)$ and
$q_1(v) = \beta(v) q_1(0\cdots 0)$ for $\beta(v)\in (0,1)$ that is independent
of $\xi_0$ and $\xi_1$.
Finally, for any choice of~$v$ satisfying $\abs{v}_1 \geq K$, we have that
$q_0(v)$ and $q_1(v)$ are both negligible.
It therefore follows that the difference $\abs{q_0(v)-q_1(v)}$ is negligible
in all cases, which completes the proof.

%-----------------------------------------------------------------------------%
\section{Conclusion}
\label{sec:con}
%-----------------------------------------------------------------------------%

This paper gives a zero-knowledge proof system for any problem in \class{QMA}
assuming the existence of a quantum computationally concealing and
unconditionally binding commitment scheme.
Such a commitment scheme can be obtained assuming quantum-secure one-way
permutations~\cite{AC02} (or injections more generally) or a quantum-secure
pseudo-random generator~\cite{Nao91} that could potentially be based on one-way
functions that are hard to invert for any quantum polynomial time
algorithm~\cite{HILL99,Zha12,Son14}.
We conclude with a few open questions and directions for future work.

\begin{mylist}{\parindent}
\item[1.]
  Our proof system inherits the soundness error of the most straightforward
  verification procedure for the local Clifford-Hamiltonian problem, which
  is to randomly select a Hamiltonian term and perform a measurement
  corresponding to it.
  When an arbitrary \class{QMA} problem is reduced to the local Hamiltonian
  problem, the resulting soundness error may potentially be large
  (polynomially bounded away from 1).
  Can one obtain a zero-knowledge proof system for any \class{QMA} problem
  with small soundness error while maintaining the other features of our
  proof system (e.g., constant round of communications)?

  We note that if a prover has polynomially many copies of a valid quantum
  witness, then a parallel repetition of our proof system may yield a
  constant round zero-knowledge proof system having small soundness error
  for any \class{QMA} problem---but this would require a parallel repetition
  result concerning zero-knowledge proof systems for \class{NP} secure
  against quantum attacks.
  Analogous results for zero-knowledge proofs for \class{NP} against
  classical attacks are known~\cite{GK96,FS89}, but they involve
  sophisticated rewinding arguments for which known quantum rewinding
  techniques do not seem to be applicable.

\item[2.]
  Are there natural formalizations of \emph{proofs of quantum knowledge}?
  Roughly speaking, one would expect such a notion to require that
  whenever a prover is able to prove the validity of a statement, one could
  construct a knowledge extractor that can extract a quantum witness given
  access to such a prover.
  It seems plausible that our proof system could be adapted to such a notion,
  although we have not investigated this notion in depth.

\item[3.]
  We have considered an encoding scheme for quantum states that ensures the
  secrecy of the state and allows for the transversal application of
  constant-size Clifford operations and measurement in the computational
  basis.
  It is an interesting open question to extend our encoding scheme, or to
  design a new one, so that it can support transversally applying a larger
  family of quantum operations.

\item[4.]
  Finally, we make one further remark on an abstract view of our proof system.
  Classically speaking, one can imagine a ``commit-and-open''
  primitive where a sender commits to a message $m$, and later opens
  sufficient information so that a receiver can test a property
  $\mathcal{P}(\cdot)$ on $m$, and nothing more.
  For example, $\mathcal{P}$ can be an \class{NP}-relation $R(x,\cdot)$ that
  checks if message $m$ is a valid witness.
  This can be implemented easily by a standard commitment scheme and during the
  opening phase, the sender and receiver run a zero-knowledge proof of
  $R(x,m) = 1$ instead of the standard opening.
  Our proof system, which combines a commitment scheme and classical
  zero-knowledge proofs for \class{NP}, can be viewed as a quantum analogue.
  Namely, we commit to a witness state and open just enough information to
  verify that some reduced density of the witness state falls into a specific
  subspace.
  We can only deal with properties of a very special form, and it is an
  interesting direction for future work to generalize and find applications of
  this sort of primitive.

\end{mylist}

%-----------------------------------------------------------------------------%
\subsection*{Acknowledgments}
%-----------------------------------------------------------------------------%

We thank Michael Beverland, Sevag Gharibian, David Gosset, Yi-Kai Liu and
Bei Zeng for helpful conversations.
A.{\,}B. and J.{\,}W. are supported in part by Canada's NSERC.
F.{\,}S. is supported in part by Cryptoworks21, Canada's NSERC and CIFAR.

\pagebreak

\appendix

%------------------------------------------------------------------------------%
\section{Preliminaries}
\label{sec:prelim}
%------------------------------------------------------------------------------%

This section summarizes some of the notation, definitions, and known
facts concerning quantum information and computation, cryptography,
and other topics that are used throughout the paper.  We refer
to~\cite{NC00,KSV02,Wat09a} for further details on the theory of
quantum information and computation.  Further information on classical
zero-knowledge and cryptography can be found in~\cite{Gol01,Gol04}.

\subsection{Basic terminology}

Throughout the paper we let $\Sigma = \{0,1\}$ denote the binary alphabet, and
only consider strings, promise problems, and complexity classes over this
alphabet.
For a string $x\in \Sigma^*$, $\abs{x}$ denotes its length.
A function $g:\mathbb{N}\rightarrow \mathbb{N}$ is a
\emph{polynomially bounded function} if there exists a deterministic
polynomial-time Turing machine $M_g$ that outputs $1^{g(n)}$ on
input~$1^n$ for every non-negative integer~$n$.
A function $f: \mathbb{N}\rightarrow [0,\infty)$ is said to be
\emph{negligible} if, for every polynomially bounded function~$g$, it holds
that $f(n) < 1/g(n)$ for all but finitely many values of~$n$.

\subsection{Quantum information basics}

When we refer to a \emph{quantum register} in this paper, we simply mean a
collection of qubits that we wish to view as a single unit and to which we give
some name.
Names of registers will always be uppercase letters in a \emph{sans serif}
font, such as $\reg{X}$, $\reg{Y}$, and $\reg{Z}$.
The finite dimensional complex Hilbert spaces associated with registers will be
denoted by capital script letters such as $\X$, $\Y$, and $\Z$, using the
same letter in the two different fonts to denote a quantum register and its
corresponding space for convenience.
Dirac notation is used to express vectors in Hilbert spaces and linear mappings
between them in a standard way.

For a given space $\X$, we let $\Lin(\X)$ denote the set of all linear
mappings (or \emph{operators}) from $\X$ to itself.
The identity element of $\Lin(\X)$ is denoted $\I_{\X}$, or just as $\I$
when $\X$ can be taken as implicit.
The inner product between operators $A$ and $B$ is defined as $\ip{A}{B} =
\tr(A^{\ast} B)$.

\emph{Quantum states} are represented by density operators, which are positive
semidefinite operators having unit trace.
A linear map $\Phi:\Lin(\X)\rightarrow\Lin(\Y)$ is said to be a
\emph{channel} if it is both completely positive and trace-preserving.
Channels are mappings from density operators to density operators
that, in principle, represent physically realizable operations.
A \emph{measurement} is described by a collections of positive semidefinite
operators $\{M_j\}$ such that $\sum_j M_j = \I$, with the probability
that the measurement on state $\rho$ results in outcome $j$ being given by
$\ip{M_j}{\rho}$

We review a few definitions of norms on operators, which are used to discuss
the distinguishability of quantum states and channels.
The \emph{trace norm} of an operator $X \in \Lin(\X)$ is defined as
$\norm{X}_1 = \tr \sqrt{X^{\ast} X}$.
For any linear map $\Phi: \Lin(\X) \rightarrow \Lin(\Y)$, the \emph{diamond
  norm} (or completely bounded trace norm)~\cite{Kit97,KSV02,AKN98} is
defined as
\begin{equation*}
  \norm{\Phi}_{\diamond} = \max\left\{ \norm{(\Phi\otimes
    \I_{\Lin(\W)})(X)}_1\,:\, X\in\Lin(\X\otimes\W)\:,\:\norm{X}_1
  \leq 1\right\},
\end{equation*}
where $\W$ is any space with dimension equal to that of $\X$.
(The value remains the same for any choice of $\W$, provided its dimension
is at least that of $\X$.)

\subsubsection*{Quantum gates and circuits}

A \emph{quantum circuit} is an acyclic network of quantum gates connected by
wires.
The quantum gates represent quantum channels while the wires represent qubits
on which the channels act.

We will refer to two types of quantum circuits in this paper:
\emph{unitary} quantum circuits and \emph{general} quantum circuits.
By unitary quantum circuits we mean circuits composed of unitary gates (such as
the ones described below) chosen from some finite gate set.
General quantum circuits are composed of gates that may correspond to channels
that are not necessarily unitary.
It is sufficient for the purposes of this paper that we consider just two simple
non-unitary gates:
\emph{ancillary gates}, which input nothing and output a qubit in the $\ket{0}$
state; and \emph{erasure gates}, which input one qubit and output nothing
(and correspond to the channel described by the trace mapping).
As is described elsewhere~\cite{AKN98, Wat11}, arbitrary channels
mapping one register to another can always be approximated arbitrarily closely
by quantum circuits whose gates include a universal collection of unitary gates
together with ancillary and erasure gates.
The \emph{size} of a quantum circuit is the number of gates in the circuit plus
the number of qubits on which it acts.

We will refer to the following well-known single-qubit unitary gates:
\begin{mylist}{\parindent}
\item[1.]
  \emph{Pauli gates:}
  \begin{equation}
    X: \ket{a} \mapsto \ket{1-a}
    \qquad\text{and}\qquad
    Z: \ket{a} \mapsto (-1)^a\ket{a},
  \end{equation}
  for each $a\in\{0,1\}$, as well as $Y = iXZ$.
\item[2.]
  \emph{Hadamard gate:}
  \begin{equation}
    H: \ket{a} \mapsto \frac{1}{\sqrt{2}}\ket{0}+\frac{(-1)^a}{\sqrt{2}}\ket{1},
  \end{equation}
  for each $a\in\{0,1\}$.
\item[3.]
  \emph{Phase gate:}
  \begin{equation}
    P: \ket{a} \mapsto i^{a}\ket{a},
  \end{equation}
  for each $a\in\{0,1\}$.
\end{mylist}
In addition, for any $k$-qubit unitary quantum gate $U$ we define the
\emph{controlled-$U$} gate as
\begin{equation}
  \Lambda(U) : \ket{a}\ket{x} \mapsto \ket{a} U^{a}\ket{x},
\end{equation}
for each $a\in\{0,1\}$ and $x\in\{0,1\}^k$.

The $k$-qubit \emph{Pauli group} is the group containing all unitary operators
of the form
\begin{equation}
  \alpha U_1\otimes \cdots \otimes U_k
\end{equation}
where $\alpha \in \{1,i,-1,-i\}$ and $U_1,\ldots,U_k\in\{\I,X,Y,Z\}$, where
$\I$ denotes the single-qubit identity operation.
Elements of this group are also referred to as \emph{Pauli operations}.
If $a,b\in\{0,1\}^k$ are binary strings of length $k$, then we write
\begin{equation}
  X^a = X^{a_1} \otimes \cdots \otimes X^{a_k}
  \quad\text{and}\quad
  Z^b = Z^{b_1} \otimes \cdots \otimes Z^{b_k}
\end{equation}
to denote the Pauli operations obtained from these strings as indicated.

Channels that can be expressed as convex combinations of unitary channels that
correspond to Pauli operations are called \emph{Pauli channels}.
An example of Pauli channels that is relevant to this paper is
the \emph{completely depolarizing} channel
\begin{equation}
  \Omega(\rho) = \frac{1}{4}\sum_{a,b\in\{0,1\}}
  \bigl(X^a Z^b\bigr) \rho \bigl(X^a Z^b\bigr)^{\ast}
  = \frac{\I}{2},
\end{equation}
for any single-qubit density operator $\rho$.
We thus see that the effect of $\Omega$ is to completely randomize the state
of a single-qubit system.
By treating a random choice of a pair $(a,b)$ as a secret key, we obtain a
quantum generalization of the one-time pad, known as the
\emph{quantum one-time pad}~\cite{AMTW00}.
When the channel is performed independently on~$k$ qubits, the effect is
given by
\begin{equation}
  \Omega^{\otimes k}(\rho) = 2^{-k}\,\I\otimes\cdots\otimes\I
\end{equation}
for every $k$-qubit density operator $\rho$.
The quantum one-time pad generalizes naturally to any choice of the number $k$.

Sometimes it will be convenient to consider quantum circuits that
implement measurements.
When we refer to a \emph{measurement circuit}, we mean any general quantum
circuit, followed by a measurement of all of its output qubits with respect to
the standard basis.
If $Q$ is a measurement circuit that is applied to a collection of
qubits in the state $\rho$, then $Q(\rho)$ is interpreted as a string-valued
random variable describing the resulting measurement.
We will only need to refer to measurement circuits outputting a single bit
in this paper.

A $k$-qubit \emph{Clifford circuit} is any unitary quantum circuit on $k$ qubits
whose gates are drawn from the set $\{H,P,\Lambda(X)\}$ containing Hadamard,
phase, and controlled-not gates.
(It is common that one also allows Pauli gates to be included in this set for
convenience.
Given that $X = HPPH$ and $Z = PP$, there is no generality lost in using the
smaller gate set in the definition.)
The set of all unitary operators that can be described by $k$-qubit Clifford
circuits forms a finite group known as the \emph{Clifford group}.
Up to scalar multiples, the $k$-qubit Clifford group is the normalizer of the
$k$-qubit Pauli group: if $U$ is a $k$-qubit unitary operator for which it
holds that $U V U^{\ast}$ is an element of the $k$-qubit Pauli group for every
$k$-qubit Pauli group element $V$, then $U = \alpha C$ for $\alpha\in\complex$
satisfying $\abs{\alpha} = 1$ and $C$ being a $k$-qubit Clifford group
element.
Given the description of a $k$-qubit Pauli group element $V$ and a $k$-qubit
Clifford circuit $C$, one can efficiently compute a description of the
$k$-qubit Pauli group element $C V C^{\ast}$~\cite{Got98}.

Clifford circuits are not universal for quantum computation.
Two examples (among other known examples) of universal gate sets are the
following:
\begin{mylist}{\parindent}
\item[1.]
  Hadamard, phase, and Toffoli gates:
  $\{H,P,\Lambda(\Lambda(X))\}$.
\item[2.]
  Hadamard and controlled-phase gates:
  $\{H,\Lambda(P)\}$.
\end{mylist}
The first of these choices is sometimes easier to work with, but we will make
use of the fact that the second gate set is universal in the paper.

\subsection{Polynomial-time generated families of quantum circuits and QMA}

Any quantum circuit with gates drawn from a fixed, finite gate set can be
encoded as a binary string, with respect to a variety of possible encoding
schemes.
The specific details of such encoding schemes are not important within the
context of this paper, so we will leave it to the reader to imagine that a
sensible and efficient encoding scheme for quantum circuits has been selected,
relative to whatever gate set is under consideration.
It should be assumed, of course, that a circuit's size and its encoding length
are polynomially related.

For any infinite set of binary strings $S\subseteq\{0,1\}^{\ast}$, a collection
$\{V_x\,:\,x\in S\}$ of quantum circuits is said to be
\emph{polynomial-time generated} if there exists a deterministic polynomial-time
Turing machine that, on input $x\in S$, outputs an encoding of $V_x$.
The assumptions on encoding schemes suggested above imply that, if
$\{V_x\,:\,x\in S\}$ is a polynomial-time generated collection, then $V_x$ must
have size polynomial in~$|x|$.

Next we will define the complexity class \class{QMA}, which is commonly viewed
as the most natural quantum generalization of \class{NP}.

\begin{definition}
  A promise problem $A = (A_{\yes},A_{\no})$ is contained in the complexity
  class $\class{QMA}_{\alpha,\beta}$ if there exists a polynomial-time generated
  collection
  \begin{equation}
    \bigl\{V_x\,:\,x\in A_{\yes} \cup A_{\no}\bigr\}
  \end{equation}
  of quantum circuits and a polynomially bounded function $p$ possessing the
  following properties:
  \begin{mylist}{\parindent}
  \item[1.]
    For every string $x\in A_{\yes}\cup A_{\no}$, one has that $V_x$ is a
    measurement circuit taking $p(\abs{x})$ input qubits and outputting
    a single bit.
  \item[2.]
    \emph{Completeness}.
    For all $x\in A_{\yes}$, there exists a $p(\abs{x})$-qubit state $\rho$
    such that $\Pr(V_x(\rho) = 1) \geq \alpha$.
  \item[3.]
    \emph{Soundness}.
    For all $x\in A_{\no}$, and every $p(\abs{x})$-qubit state
    $\rho$, it holds that $\Pr(V_x(\rho) = 1) \leq \beta$.
  \end{mylist}
\end{definition}

\noindent
In this definition, $\alpha,\beta\in[0,1]$ may be constant values or functions
of the length of the input string~$x$.
When they are omitted, it is to be assumed that they are $\alpha = 2/3$ and
$\beta = 1/3$.
Known error reduction methods~\cite{KSV02,MW05} imply that a wide
range of selections of $\alpha$ and $\beta$ give rise to the same complexity
class.
In particular, \class{QMA} coincides with $\class{QMA}_{\alpha,\beta}$ for
$\alpha = 1 - 2^{-q(\abs{x})}$ and $\beta = 2^{-q(\abs{x})}$, for any
polynomially bounded function $q$.

\subsection{Quantum computational indistinguishability and zero-knowledge}

Next we review notions of quantum state and channel discrimination, as
well as zero-knowledge in a quantum setting (as defined in~\cite{Wat09}).

We first specify what it means for two collections of quantum states to be
quantum computationally indistinguishable.
The definition that follows may be viewed as being a non-uniform notion of
quantum computational indistinguishability, as it places no uniformity
conditions on quantum circuits and allows for an \emph{auxiliary} quantum
state $\sigma$ to assist in the task of state discrimination.

\begin{definition}[Quantum computationally indistinguishable states]
  \label{def_ind_states}
  Suppose that $S\subseteq\{0,1\}^{\ast}$ is an infinite set of binary strings,
  $r$ is a polynomially bounded function, and $\rho_x$ and $\xi_x$ are states on
  $r(\abs{x})$ qubits for each $x\in S$.
  The collections $\{\rho_x\,:\,x\in S\}$ and $\{\xi_x\,:\,x\in S\}$ are
  \emph{quantum computationally indistinguishable} if, for every choice of
  polynomially bounded functions $s$ and $k$, any measurement circuit $Q$ of
  size $s(\abs{x})$, and any choice of a $k(\abs{x})$-qubit state $\sigma$, it
  holds that
  \begin{equation}
    \abs{\Pr[Q(\rho_x \otimes\sigma) = 1] -
      \Pr[Q(\xi_x \otimes\sigma) = 1]} \leq \varepsilon(\abs{x})
  \end{equation}
  for all $x\in S$, for a negligible function $\varepsilon$.
\end{definition}

The notion extends naturally to distinguishing collections of channels, as the
following definition makes precise.

\begin{definition}[Quantum computationally indistinguishable channels]
  Suppose that $S\subseteq\{0,1\}^{\ast}$ is an infinite set of binary strings,
  $q$ and $r$ are polynomially bounded functions, and $\Phi_x$ and
  $\Psi_x$ are channels from $q(\abs{x})$ qubits to $r(\abs{x})$ qubits for
  each $x\in S$.
  The collections $\{\Phi_x\,:\, x \in S\}$ and
  $\{\Psi_x\,:\,x \in S\}$
  are \emph{quantum computationally indistinguishable} if, for every choice of
  polynomially bounded functions $s$ and $k$, every state $\sigma$ on
  $q(\abs{x})+k(\abs{x})$ qubits, and every measurement circuit $Q$ on
  $r(\abs{x})+k(\abs{x})$ qubits having size $s(\abs{x})$, one has that
  \begin{equation}
  \abs{\Pr[ Q( (\Phi_x \otimes \I)(\sigma)) = 1] -\Pr[ Q((\Psi_x
    \otimes \I)(\sigma)) = 1] } \leq \varepsilon(\abs{x})
  \end{equation}
  for every $x\in S$, for a negligible function $\varepsilon$.
  \label{def_ind_so}
\end{definition}

We will also make use of statistical notions of indistinguishability for states
and channels, which are defined as follows.

\begin{definition}[Statistically indistinguishable states]
  Suppose that $S\subseteq\{0,1\}^{\ast}$ is an infinite set of binary strings,
  $r$ is a polynomially bounded function, and $\rho_x$ and $\xi_x$ are states on
  $r(\abs{x})$ qubits for each $x\in S$.
  The collections $\{\rho_x\,:\,x\in S\}$ and
  $\{\xi_x\,:\,x\in S\}$ are \emph{statistically indistinguishable} if
  \begin{equation}
    \frac{1}{2} \norm{\rho_x - \xi_x}_1 \leq \varepsilon(|x|) \,
  \end{equation}
  for all $x\in S$, for a negligible function $\varepsilon$.
  \label{def_sind_states}
\end{definition}

\begin{definition}[Statistically indistinguishable channels]
  Suppose that $S\subseteq\{0,1\}^{\ast}$ is an infinite set of binary strings,
  $q$ and $r$ are polynomially bounded functions, and $\Phi_x$ and
  $\Psi_x$ are channels from $q(\abs{x})$ qubits to $r(\abs{x})$ qubits for
  each $x\in S$.
  The collections $\{\Phi_x\,:\,x\in S\}$ and
  $\{\Psi_x\,:\,x\in S\}$ are \emph{statistically indistinguishable} if
  \begin{equation}
    \frac{1}{2}\norm{\Phi_x - \Psi_x}_{\diamond} \leq \varepsilon(\abs{x})
  \end{equation}
  for all $x\in S$, for a negligible function $\varepsilon$.
\end{definition}

Next we review the definition of quantum computational zero-knowledge
proof systems as defined in~\cite{Wat09}.
Let $(P,V)$ be a quantum or classical interactive proof system for a promise
problem $A$.
An arbitrary (possibly malicious) verifier $V'$ is any quantum computational
process that interacts with $P$ according to the structural specification of
$(P,V)$.
Similar to the classical notion of auxiliary input zero-knowledge, a verifier
$V'$ will take, in addition to the input string $x$, an auxiliary input, and
produce some output.
This is crucial for the composition of zero-knowledge proof systems.
The most general situation allowed by quantum information theory is that both
the auxiliary input and the output are quantum, meaning that the verifier
operates on quantum registers whose initial state is arbitrary and may be
entangled with some external system.
Also similar to the classical case, we will assume that for any given
polynomial-time verifier $V'$ there exist polynomially bounded functions $q$
and $r$ that determine the number of auxiliary input qubits and output qubits
of $V'$.
To say that $V'$ is a polynomial-time verifier means that the entire action of
$V'$ must be described by some polynomial-time generated family of quantum
circuits.

The interaction of a verifier $V'$ with $P$ on input~$x$ induces some channel
from the verifier's $q(\abs{x})$ auxiliary input qubits to $r(\abs{x})$ output
qubits.
Let $\W$ denote the vector space corresponding to the auxiliary input qubits,
let $\Z$ denote the space corresponding to the output qubits, and let
$\Phi_x: \Lin(\W)\rightarrow \Lin(\Z)$ denote the resulting channel induced by
the interaction of $V'$ with $P$ on input~$x$.
A simulator $S$ for a given verifier $V'$ is described by a polynomial-time
generated family of general quantum circuits that agrees with $V'$ on the
functions $q$ and $r$ representing the number of auxiliary input qubits and
output qubits respectively.
Such a simulator does not interact with~$P$, but simply induces a channel that
we will denote by $\Psi_x : \Lin(\W) \rightarrow \Lin(\Z)$ on each input $x$.

\begin{definition}[Quantum computational zero-knowledge]
  An interactive proof system $(P, V)$ for a promise problem $A$ is
  \emph{quantum computational zero-knowledge} if, for every polynomial-time
  generated quantum verifier $V'$, there exists a polynomial-time generated
  quantum simulator $S$ that satisfies the following requirements.
  \begin{mylist}{\parindent}
  \item[1.]
    The verifier $V'$ and simulator $S$ agree on the polynomially bounded
    functions $q$ and $r$ that specify the number of auxiliary input qubits and
    output qubits, respectively.
  \item[2.]
    Let $\Phi_x$ be the channel that results from the interaction between $V'$
    and $P$ on input~$x$, and let $\Psi_x$ be the channel induced by the
    simulator $S$ on input~$x$, both as described above.
    Then the collections $\{\Phi_x: x\in A_{\yes}\}$ and
    $\{\Psi_x: x \in A_{\yes}\}$ are quantum computationally
    indistinguishable.
  \end{mylist}
  \label{def_qczk}
\end{definition}

\subsection{Cryptographic Tools}

Here we introduce a few cryptographic building blocks that are useful
in our proof system.
We emphasize that, as is typical in the classical setting, we formulate all
computational security properties (e.g., concealing in a commitment scheme)
with respect to non-uniform quantum adversaries.
This is inherited from the definition of quantum computational
indistinguishability.
This gives more stringent security requirements and is also crucial in security
proofs.

\subsubsection*{Commitment schemes}

For the sake of simplicity, we describe a commitment scheme that is
non-interactive, i.e., all messages are going from a sender to a
receiver. A similar definition can be derived for interactive schemes.

\begin{definition}[Quantum computationally secure commitment schemes]
  A \emph{quantum computationally secure commitment scheme} for an alphabet
  $\Gamma$ is a collection of polynomial-time computable functions
  $\{f_n\,:\,n\in\natural\}$ taking the form
  \begin{equation}
    f_n:\Gamma\times\{0,1\}^{p(n)} \rightarrow \{0,1\}^{q(n)},
  \end{equation}
  for polynomially bounded functions $p$ and $q$, such that the following
  conditions hold:
  \begin{mylist}{\parindent}
  \item[1.]
    \emph{Unconditionally binding property.}
    For every choice of $n\in\natural$, $a,b\in\Gamma$, and
    $r,s\in\{0,1\}^{p(n)}$, one has that $f_n(a,r) = f_n(b,s)$ implies $a=b$.
  \item[2.]
    \emph{Quantum computationally concealing property.}
    For every $a\in\Gamma$ and $n\in\natural$, define
    \begin{equation}
      \rho_{a,n} = \frac{1}{2^{p(n)}}\sum_{r\in\{0,1\}^{p(n)}}
      \ket{f_n(a,r)}\bra{f_n(a,r)}.
    \end{equation}
    For every choice of $a,b\in\Gamma$ the ensembles
    $\{\rho_{a,n}\,:\,n\in\natural\}$ and $\{\rho_{b,n}\,:\,n\in\natural\}$
    are quantum computationally indistinguishable.
  \end{mylist}
\end{definition}

To commit to a string, one can independently use the commitment
described above bit by bit. Such a commitment scheme can be
constructed based on certain quantum intractability assumptions. As
shown in~\cite{AC02}, it suffices to have quantum-resistant one-way
\emph{permutations}, which are permutations that can be computed
efficiently on a classical computer but are hard to invert for both
classical and quantum polynomial-time algorithms. The same commitment
scheme remains quantum-secure based on a slightly weaker assumption of
quantum-resistant \emph{injective} one-way functions.
% If we are willing to pay for a two-message commit phase,
Naor showed a commitment scheme with a two-message commit
phase~\cite{Nao91} which will be quantum-secure~\cite{HSS15}, assuming
one uses a pseudo-random generator whose output is \emph{quantum}
computationally indistinguishable from a truly random
string\footnote{It has been stated informally (see
  e.g.,~\cite{Zha12,Son14}) that the pseudo-random generator by
  H{\aa}stad et al.~\cite{HILL99} based on one-way \emph{functions}
  would remain quantum-secure, so long as the one-way functions are
  resistant to any polynomial-time quantum inverting algorithms.}.

Based on such a quantum-secure commitment scheme, we can obtain the
other two essential cryptographic building blocks in our protocol: a
zero-knowledge proof system for \class{NP} and a coin-flipping
protocol, both secure against quantum adversaries.

\subsubsection*{Zero-knowledge proof for \class{NP}} Watrous showed
that~\cite{Wat09} the GMW $3$-Coloring protocol~\cite{GMW91} remains
zero-knowledge in the presence of quantum verifiers, assuming a
statistically binding and quantum computationally hiding commitment
scheme. This means that we have a classical zero-knowledge proof
protocol for any \class{NP} language that is secure against any
polynomial-time quantum verifiers.

\subsubsection*{Coin-flipping}

A coin-flipping protocol is an interactive process that allows two parties
to jointly toss random coins.
It is not necessary for us to consider this notion generally, as we only make
use of one specific coin-flipping protocol, namely Blum's coin-flipping
protocol~\cite{Blu83} in which an honest prover commits to a random
$y\in\{0,1\}$, the honest verifier selects $z\in\{0,1\}$ at random, the prover
reveals~$y$, and the two participants agree that the random bit generated
$r = y\oplus z$.

Damg{\aa}rd and Lunemann~\cite{DL09} proved that Blum's coin-flipping
protocol is quantum-secure, assuming a quantum-secure commitment scheme.
This protocol generates one random coin, and we will need to flip logarithmic
many random bits.
A simple way of achieving this is by sequential repetition, but more
effectively it is possible to extend the analysis of Damg{\aa}rd and Lunemann
and show that parallel repetition of Blum's protocol logarithmic many times
remains quantum-secure.

\subsection{Concatenated Steane codes}
\label{sec:Steane-code}

The last topic to be discussed in this section concerns the existence of
quantum error correcting codes having certain properties that are important to
the functioning of our zero-knowledge proof system for \class{QMA}.
There are multiple choices of codes that satisfy our requirements, but in the
interest of simplicity we will describe just one specific family of codes
in this category.

These codes are based on the \emph{7-qubit Steane code}~\cite{Ste96}, in which
one qubit is encoded into 7 qubits by the following action on standard basis
states:
\begin{equation}
  \ket{0} \mapsto \frac{1}{\sqrt{8}} \sum_{x\in\D_7^0} \ket{x}
  \qquad\text{and}\qquad
  \ket{1} \mapsto \frac{1}{\sqrt{8}} \sum_{x\in\D_7^1} \ket{x},
\end{equation}
where
\begin{equation}
  \begin{aligned}
    \D_7^0 & =
    \{0000000,0001111,0110011,0111100,1010101,1011010,1100110,1101001\},\\
    \D_7^1 & =
    \{0010110,0011001,0100101,0101010,1000011,1001100,1110000,1111111\}.
  \end{aligned}
\end{equation}
It is the case that $\D_7^0$ is a $[7,4]$-Hamming code, while
\begin{equation}
  \D_7 = \D_7^0 \cup \D_7^1
\end{equation}
is the dual code to $\D_7^0$ (i.e., it is the code consisting of all binary
strings of length 7 whose inner product with any codeword in $\D_7^0$ is
even).
This is an example of a \emph{CSS code}~\cite{NC00}, and it is capable of
correcting single-qubit errors.
The standard error-correcting procedure, which we do not actually need in this
paper, is to first reversibly correct errors in the standard basis, with respect
to the code $\D_7$, and then to do the same with respect to the diagonal basis.
The 7-qubit Clifford circuit depicted in Figure~\ref{fig:Steane-encoder}
encodes one qubit into 7 with respect to this code, assuming 6 qubits in
the~$\ket{0}$ state are made available.

\begin{figure}
  \begin{mdframed}[style=figstyle]
  \begin{center}
    \begin{tikzpicture}[scale=0.7,
        control/.style={circle, fill, minimum size = 4pt, inner sep=0mm},
        target/.style={circle, draw, minimum size = 7pt, inner sep=0mm},
        gate/.style={draw, fill = ChannelColor, minimum size = 16pt}]

      \node (In0) at (-6.5,3) {$\ket{\psi}$};
      \node (In1) at (-6.5,2) {$\ket{0}$};
      \node (In2) at (-6.5,1) {$\ket{0}$};
      \node (In3) at (-6.5,0) {$\ket{0}$};
      \node (In4) at (-6.5,-1) {$\ket{0}$};
      \node (In5) at (-6.5,-2) {$\ket{0}$};
      \node (In6) at (-6.5,-3) {$\ket{0}$};

      \node (Out0) at (7,3) {};
      \node (Out1) at (7,2) {};
      \node (Out2) at (7,1) {};
      \node (Out3) at (7,0) {};
      \node (Out4) at (7,-1) {};
      \node (Out5) at (7,-2) {};
      \node (Out6) at (7,-3) {};

      \node[gate] (H4) at (-5,-1) {$H$};
      \node[gate] (H5) at (-5,-2) {$H$};
      \node[gate] (H6) at (-5,-3) {$H$};

      \draw (In0) -- (Out0) {};
      \draw (In1) -- (Out1) {};
      \draw (In2) -- (Out2) {};
      \draw (In3) -- (Out3) {};
      \draw (In4) -- (H4) -- (Out4) {};
      \draw (In5) -- (H5) -- (Out5) {};
      \draw (In6) -- (H6) -- (Out6) {};

      \node[control] (Control01) at (-4,3) {};
      \node[target] (Target11) at (-4,2) {};
      \draw (Control01.center) -- (Target11.south);

      \node[control] (Control02) at (-3,3) {};
      \node[target] (Target22) at (-3,1) {};
      \draw (Control02.center) -- (Target22.south);

      \node[control] (Control63) at (-2,-3) {};
      \node[target] (Target33) at (-2,0) {};
      \draw (Control63.center) -- (Target33.north);

      \node[control] (Control64) at (-1,-3) {};
      \node[target] (Target14) at (-1,2) {};
      \draw (Control64.center) -- (Target14.north);

      \node[control] (Control65) at (0,-3) {};
      \node[target] (Target05) at (0,3) {};
      \draw (Control65.center) -- (Target05.north);

      \node[control] (Control66) at (1,-2) {};
      \node[target] (Target36) at (1,0) {};
      \draw (Control66.center) -- (Target36.north);

      \node[control] (Control67) at (2,-2) {};
      \node[target] (Target17) at (2,1) {};
      \draw (Control67.center) -- (Target17.north);

      \node[control] (Control68) at (3,-2) {};
      \node[target] (Target08) at (3,3) {};
      \draw (Control68.center) -- (Target08.north);

      \node[control] (Control69) at (4,-1) {};
      \node[target] (Target39) at (4,0) {};
      \draw (Control69.center) -- (Target39.north);

      \node[control] (Control6a) at (5,-1) {};
      \node[target] (Target1a) at (5,1) {};
      \draw (Control6a.center) -- (Target1a.north);

      \node[control] (Control6b) at (6,-1) {};
      \node[target] (Target0b) at (6,2) {};
      \draw (Control6b.center) -- (Target0b.north);

    \end{tikzpicture}
  \end{center}
  \caption{A Clifford circuit encoder for the $7$-qubit Steane code.
    Hereafter we will write $U_7$ to refer to the unitary operator on 7 qubits
    described by this circuit.}
  \label{fig:Steane-encoder}
\end{mdframed}
\end{figure}
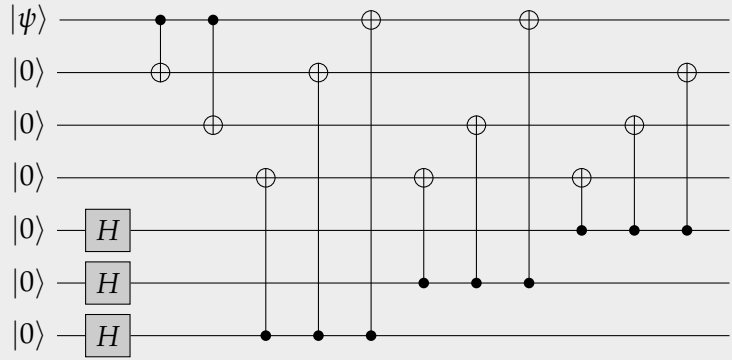

One of the properties of the 7-qubit Steane code that is important from the
viewpoint of this paper is that it admits a \emph{transversal} application of
Clifford operations, in the sense that is explained in
Figure~\ref{fig:Clifford-Steane-transversal}.

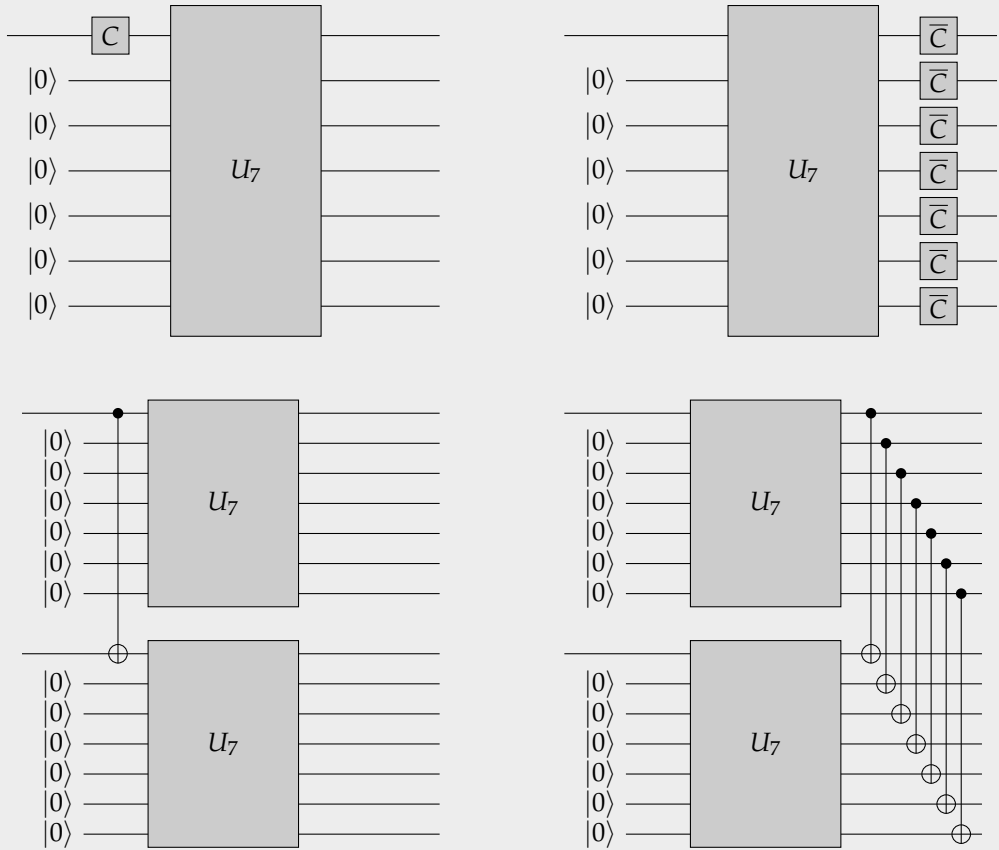
\begin{figure}[t]
  \begin{mdframed}[style=figstyle]
  \begin{center}\small
    \begin{tikzpicture}[scale=0.6,
        control/.style={circle, fill, minimum size = 4pt, inner sep=0mm},
        target/.style={circle, draw, minimum size = 7pt, inner sep=0mm},
        gate/.style={draw, minimum size = 14pt, fill = ChannelColor, inner
          sep=2pt},
        biggate/.style={draw, minimum height = 4.4cm, minimum width = 2cm,
          fill=ChannelColor}]

      \node (In0) at (-7.5,3) {};
      \node (In1) at (-6.5,2) {$\ket{0}$};
      \node (In2) at (-6.5,1) {$\ket{0}$};
      \node (In3) at (-6.5,0) {$\ket{0}$};
      \node (In4) at (-6.5,-1) {$\ket{0}$};
      \node (In5) at (-6.5,-2) {$\ket{0}$};
      \node (In6) at (-6.5,-3) {$\ket{0}$};

      \node (Out0) at (2.5,3) {};
      \node (Out1) at (2.5,2) {};
      \node (Out2) at (2.5,1) {};
      \node (Out3) at (2.5,0) {};
      \node (Out4) at (2.5,-1) {};
      \node (Out5) at (2.5,-2) {};
      \node (Out6) at (2.5,-3) {};

      \draw (In0) -- (Out0);
      \draw (In1) -- (Out1);
      \draw (In2) -- (Out2);
      \draw (In3) -- (Out3);
      \draw (In4) -- (Out4);
      \draw (In5) -- (Out5);
      \draw (In6) -- (Out6);

      \node[gate] (H00) at (-5,3) {$C$};
      \node[biggate] (V) at (-2,0) {$U_7$};

    \end{tikzpicture}
    \hspace*{1cm}
    \begin{tikzpicture}[scale=0.6,
        control/.style={circle, fill, minimum size = 4pt, inner sep=0mm},
        target/.style={circle, draw, minimum size = 7pt, inner sep=0mm},
        gate/.style={draw, minimum size = 14pt, fill = ChannelColor, inner
          sep=2pt},
        biggate/.style={draw, minimum height = 4.4cm, minimum width = 2cm,
          fill=ChannelColor}]

      \node (In0) at (-7.5,3) {};
      \node (In1) at (-6.5,2) {$\ket{0}$};
      \node (In2) at (-6.5,1) {$\ket{0}$};
      \node (In3) at (-6.5,0) {$\ket{0}$};
      \node (In4) at (-6.5,-1) {$\ket{0}$};
      \node (In5) at (-6.5,-2) {$\ket{0}$};
      \node (In6) at (-6.5,-3) {$\ket{0}$};

      \node (Out0) at (2.5,3) {};
      \node (Out1) at (2.5,2) {};
      \node (Out2) at (2.5,1) {};
      \node (Out3) at (2.5,0) {};
      \node (Out4) at (2.5,-1) {};
      \node (Out5) at (2.5,-2) {};
      \node (Out6) at (2.5,-3) {};

      \draw (In0) -- (Out0);
      \draw (In1) -- (Out1);
      \draw (In2) -- (Out2);
      \draw (In3) -- (Out3);
      \draw (In4) -- (Out4);
      \draw (In5) -- (Out5);
      \draw (In6) -- (Out6);

      \node[gate] (H01) at (1,3) {$\overline{C}$};
      \node[gate] (H11) at (1,2) {$\overline{C}$};
      \node[gate] (H21) at (1,1) {$\overline{C}$};
      \node[gate] (H31) at (1,0) {$\overline{C}$};
      \node[gate] (H41) at (1,-1) {$\overline{C}$};
      \node[gate] (H51) at (1,-2) {$\overline{C}$};
      \node[gate] (H61) at (1,-3) {$\overline{C}$};
      \node[biggate] (V) at (-2,0) {$U_7$};
    \end{tikzpicture}\\[8mm]
    \begin{tikzpicture}[scale=0.8,
        control/.style={circle, fill, minimum size = 4pt, inner sep=0mm},
        target/.style={circle, draw, minimum size = 7pt, inner sep=0mm},
        gate/.style={draw, minimum size = 14pt, fill = ChannelColor, inner
          sep=2pt},
        biggate/.style={draw, minimum height = 2.75cm, minimum width = 2cm,
          fill=ChannelColor}]

      \node (In0) at (-5.5,3.5) {};
      \node (In1) at (-4.75,3) {$\ket{0}$};
      \node (In2) at (-4.75,2.5) {$\ket{0}$};
      \node (In3) at (-4.75,2) {$\ket{0}$};
      \node (In4) at (-4.75,1.5) {$\ket{0}$};
      \node (In5) at (-4.75,1) {$\ket{0}$};
      \node (In6) at (-4.75,0.5) {$\ket{0}$};

      \node (Out0) at (1.75,3.5) {};
      \node (Out1) at (1.75,3) {};
      \node (Out2) at (1.75,2.5) {};
      \node (Out3) at (1.75,2) {};
      \node (Out4) at (1.75,1.5) {};
      \node (Out5) at (1.75,1) {};
      \node (Out6) at (1.75,0.5) {};

      \node (In0d) at (-5.5,-0.5) {};
      \node (In1d) at (-4.75,-1) {$\ket{0}$};
      \node (In2d) at (-4.75,-1.5) {$\ket{0}$};
      \node (In3d) at (-4.75,-2) {$\ket{0}$};
      \node (In4d) at (-4.75,-2.5) {$\ket{0}$};
      \node (In5d) at (-4.75,-3) {$\ket{0}$};
      \node (In6d) at (-4.75,-3.5) {$\ket{0}$};

      \node (Out0d) at (1.75,-0.5) {};
      \node (Out1d) at (1.75,-1) {};
      \node (Out2d) at (1.75,-1.5) {};
      \node (Out3d) at (1.75,-2) {};
      \node (Out4d) at (1.75,-2.5) {};
      \node (Out5d) at (1.75,-3) {};
      \node (Out6d) at (1.75,-3.5) {};

      \draw (In0) -- (Out0);
      \draw (In1) -- (Out1);
      \draw (In2) -- (Out2);
      \draw (In3) -- (Out3);
      \draw (In4) -- (Out4);
      \draw (In5) -- (Out5);
      \draw (In6) -- (Out6);

      \node[biggate] (V) at (-2,2) {$U_7$};

      \draw (In0d) -- (Out0d);
      \draw (In1d) -- (Out1d);
      \draw (In2d) -- (Out2d);
      \draw (In3d) -- (Out3d);
      \draw (In4d) -- (Out4d);
      \draw (In5d) -- (Out5d);
      \draw (In6d) -- (Out6d);

      \node[biggate] (Vd) at (-2,-2) {$U_7$};

      \node[control] (C0) at (-3.75,3.5) {};
      \node[target] (T0) at (-3.75,-0.5) {};
      \draw (C0.center) -- (T0.south);

    \end{tikzpicture}
    \hspace*{1cm}
    \begin{tikzpicture}[scale=0.8,
        control/.style={circle, fill, minimum size = 4pt, inner sep=0mm},
        target/.style={circle, draw, minimum size = 7pt, inner sep=0mm},
        gate/.style={draw, minimum size = 14pt, fill = ChannelColor,
          inner sep=2pt},
        biggate/.style={draw, minimum height = 2.75cm, minimum width = 2cm,
          fill=ChannelColor}]

      \node (In0) at (-5.5,3.5) {};
      \node (In1) at (-4.75,3) {$\ket{0}$};
      \node (In2) at (-4.75,2.5) {$\ket{0}$};
      \node (In3) at (-4.75,2) {$\ket{0}$};
      \node (In4) at (-4.75,1.5) {$\ket{0}$};
      \node (In5) at (-4.75,1) {$\ket{0}$};
      \node (In6) at (-4.75,0.5) {$\ket{0}$};

      \node (Out0) at (1.75,3.5) {};
      \node (Out1) at (1.75,3) {};
      \node (Out2) at (1.75,2.5) {};
      \node (Out3) at (1.75,2) {};
      \node (Out4) at (1.75,1.5) {};
      \node (Out5) at (1.75,1) {};
      \node (Out6) at (1.75,0.5) {};

      \node (In0d) at (-5.5,-0.5) {};
      \node (In1d) at (-4.75,-1) {$\ket{0}$};
      \node (In2d) at (-4.75,-1.5) {$\ket{0}$};
      \node (In3d) at (-4.75,-2) {$\ket{0}$};
      \node (In4d) at (-4.75,-2.5) {$\ket{0}$};
      \node (In5d) at (-4.75,-3) {$\ket{0}$};
      \node (In6d) at (-4.75,-3.5) {$\ket{0}$};

      \node (Out0d) at (1.75,-0.5) {};
      \node (Out1d) at (1.75,-1) {};
      \node (Out2d) at (1.75,-1.5) {};
      \node (Out3d) at (1.75,-2) {};
      \node (Out4d) at (1.75,-2.5) {};
      \node (Out5d) at (1.75,-3) {};
      \node (Out6d) at (1.75,-3.5) {};

      \draw (In0) -- (Out0);
      \draw (In1) -- (Out1);
      \draw (In2) -- (Out2);
      \draw (In3) -- (Out3);
      \draw (In4) -- (Out4);
      \draw (In5) -- (Out5);
      \draw (In6) -- (Out6);

      \node[biggate] (V) at (-2,2) {$U_7$};

      \draw (In0d) -- (Out0d);
      \draw (In1d) -- (Out1d);
      \draw (In2d) -- (Out2d);
      \draw (In3d) -- (Out3d);
      \draw (In4d) -- (Out4d);
      \draw (In5d) -- (Out5d);
      \draw (In6d) -- (Out6d);

      \node[biggate] (Vd) at (-2,-2) {$U_7$};

      \node[control] (C01) at (-0.25,3.5) {};
      \node[target] (T01) at (-0.25,-0.5) {};
      \draw (C01.center) -- (T01.south);

      \node[control] (C11) at (0,3) {};
      \node[target] (T11) at (0,-1) {};
      \draw (C11.center) -- (T11.south);

      \node[control] (C21) at (0.25,2.5) {};
      \node[target] (T21) at (0.25,-1.5) {};
      \draw (C21.center) -- (T21.south);

      \node[control] (C31) at (0.5,2) {};
      \node[target] (T31) at (0.5,-2) {};
      \draw (C31.center) -- (T31.south);

      \node[control] (C41) at (0.75,1.5) {};
      \node[target] (T41) at (0.75,-2.5) {};
      \draw (C41.center) -- (T41.south);

      \node[control] (C51) at (1,1) {};
      \node[target] (T51) at (1,-3) {};
      \draw (C51.center) -- (T51.south);

      \node[control] (C61) at (1.25,0.5) {};
      \node[target] (T61) at (1.25,-3.5) {};
      \draw (C61.center) -- (T61.south);

    \end{tikzpicture}
  \end{center}
  \caption{The 7-qubit Steane code allows for the transversal application of
    Clifford operations.
    That is, the circuits on the left are equivalent to the corresponding
    circuits on the right.
    In general, the application of any Clifford operation on $k$ qubits prior
    to being encoded is equivalent to the entry-wise complex conjugate of that
    Clifford operation being applied 7 times to the $7k$ qubits that encode the
    original $k$ qubits.}
  \label{fig:Clifford-Steane-transversal}
\end{mdframed}
\end{figure}

Note that by concatenating the 7-qubit Steane code with itself, one
obtains a code having similar properties to the 7-qubit code, and in
addition having a large minimum distance for the underlying code.
More specifically, suppose that $N = 7^t$ for $t$ being an even
positive integer.  (We take $t$ to be even for convenience, as this
eliminates the entry-wise complex conjugation on Clifford operations
encountered in the discussion of their transversal application.)  By
concatenating the 7-qubit Steane code to itself $t$ times, one obtains
a quantum error-correcting code in which one qubit is encoded into $N$
qubits in the following way:
\begin{equation}
  \ket{0} \mapsto \frac{1}{\sqrt{8^t}} \sum_{x\in\D_N^0} \ket{x}
  \quad\text{and}\quad
  \ket{1} \mapsto \frac{1}{\sqrt{8^t}} \sum_{x\in\D_N^1} \ket{x}
\end{equation}
where $\D_N^0,\D_N^1\subseteq\{0,1\}^N$ are related in a way that generalizes
the case $N = 7$.
In particular, $\D_N^0$ is a binary linear code having $8^t$ elements, and
whose dual code takes the form
\begin{equation}
  \D_N = \D_N^0 \cup \D_N^1
\end{equation}
for $\D_N^1\subseteq\{0,1\}^N$ being a coset of $\D_N^0$.

As a quantum error correcting code, the $t$-fold concatenation of the 7-qubit
Steane code inherits the properties of the 7-qubit Steane code mentioned above.
A Clifford circuit $U_N$ acting on $N$ qubits, $N-1$ of which are to be
initialized in the $\ket{0}$ state, performs the encoding.
This circuit is obtained by creating a tree from multiple copies of the circuit
$U_7$ in the natural way.
The code allows for Clifford operations to be applied transversally.

An added feature of the concatenated versions of the 7-qubit Steane code is
that it corrects more errors than the ordinary 7-qubit code.
In particular, we will make use of the fact that the code $\D_N$, for
$N = 7^t$, has minimum Hamming weight $3^t$ for a nonzero code word.
This allows one to obtain a polynomial-length code for any polynomial
lower-bound on the minimum nonzero Hamming weight of a code word.

\bibliographystyle{acm}
\bibliography{QZK-QMA}

\end{document}